\documentclass[aps,pra,twocolumn,groupedaddress]{revtex4}
\usepackage{graphicx}
\usepackage{epstopdf}
\usepackage{algorithm2e}

\usepackage{amsmath,amssymb,amsfonts}
\usepackage{fullwidth}
\usepackage{capt-of}
\usepackage{pbox}
\usepackage{lipsum}
\usepackage[noend]{algpseudocode}
\usepackage{amsthm}
\usepackage{caption}
\usepackage{pseudocode}
\usepackage{multirow}
\usepackage{verbatim}
\usepackage{url}
\usepackage{comment}
\usepackage{mathtools}
\usepackage{epsf,pgf,graphicx}
\usepackage{tikz,pgf}
\usepackage{diagbox}
\usepackage{makecell}
\usepackage{color}
\usepackage{tcolorbox}
\usepackage[colorlinks=true,urlcolor=blue]{hyperref}
\newcommand{\K}{\mathcal{K}}
\newcommand{\A}{\mathcal{A}}
\newcommand{\B}{\mathcal{B}}

\newcommand{\kep}{\textit{key establishment phase}}

\newcommand{\pqp}{\textit{private query phase }}
\newcommand{\evp}{\textit{source device verification phase}}
\newcommand{\bvp}{\textit{DI testing phase for Bob's measurement device}}
\newcommand{\avp}{\textit{DI testing phase for Alice's POVM elements}}
\newcommand{\CHSH}{\text{CHSH}}
\newcommand{\obs}{\text{obs}}

\newcommand{\ind}{\mathcal{I}}
\newcommand{\tr}{\mathrm{Tr}}
\newcommand{\gs}{\mathrm{guess}}
\newcommand{\h}{\mathrm{H}}
\newcommand{\I}{\mathbb{I}}

\newcommand{\id}{\mathbb{I}}

\newcommand{\m}{\vec{m}}
\newcommand{\asig}{\vec{\sigma}}

\newtheorem{theorem}{Theorem}
\newtheorem{corollary}{Corollary}
\newtheorem{lemma}{Lemma}
\newtheorem{definition}{Definition}
\newtheorem{proposition}{Proposition}
\newcommand{\wa}{W_{\alpha}}

\newcommand{\ket}[1]{\ensuremath{\left|#1\right\rangle}}
\newcommand{\bra}[1]{\ensuremath{\left\langle#1\right|}}

\begin{document}
	\title{Improved and Formal Proposal for Device Independent Quantum Private Query} 
	\author{Jyotirmoy Basak$^1$\footnote{bjyotirmoy.93@gmail.com}, Kaushik Chakraborty$^2$\footnote{kaushik.chakraborty9@gmail.com}, Arpita Maitra$^3$\footnote{arpita76b@gmail.com}, Subhamoy Maitra$^1$\footnote{subho@isical.ac.in}}
	\affiliation{$^1$Applied Statistics Unit, Indian Statistical Institute, Kolkata, India.\\
		$^2$School of Informatics, The University of Edinburgh, UK.\\
		$^3$TCG Centre for Research and Education in Science and Technology, Kolkata, India.}

	\begin{abstract}

		In this paper, we propose a novel Quantum Private Query (QPQ) scheme with full Device-Independent certification. To the best of our knowledge, this is the first time we provide such a full DI-QPQ scheme using EPR-pairs. Our proposed scheme exploits self-testing of shared EPR-pairs along with the self-testing of projective measurement operators in a setting where the client and the server do not trust each other. To certify full device independence, we exploit a strategy to self-test a particular class of POVM elements that are used in the protocol. Further, we provide formal security analysis and obtain an upper bound on the maximum cheating probabilities for both the dishonest client as well as the dishonest server.			

	\end{abstract}

	\maketitle

\section{Introduction} 
	Since the very first proposal by Chor et al. \cite{CGKS95}, both Private Information Retrieval (PIR), and Symmetric PIR have attracted extensive attention from the classical cryptography domain.  \cite{KO97,Micali99,Gentry05,William07,Giovanni00,GIKM98}. SPIR is a two-party (say Server, and Client) mistrustful crypto primitive. Informally, in SPIR one party, Client would like to retrieve some information from a database that is stored at the other party, i.e., Server's side without revealing any information about the retrieved data bits to the Server. The Server's goal is not to reveal any information about the rest of the database. The task of SPIR is similar to the 
$1$ out of $N$ oblivious transfer. Similar to most of the secure two-party cryptographic primitives, designing a secure SPIR scheme is a difficult task. 
Since the client's privacy and the database security appear to be conflicting, it is elusive to design information-theoretically secure SPIR schemes both in classical and in quantum domain \cite{GIKM98,Lo97}. This paper focuses on a more weaker version of SPIR, called Private Query (PQ), where the client is allowed to gain more information about the database 
than SPIR or $1$ out of $N$ oblivious transfer. On the other hand, the client's privacy is ensured in the sense of cheat sensitivity i.e., if the server tries to gain the information about the client's queries then the client can detect that.

The PQ primitive is weaker than SPIR but stronger than PIR. However, this type of primitive suffers from the same limitation as in the PIR schemes. 
For example, in order to respond client's query, the server must process the entire database. 
Otherwise, the server will gain  information regarding the indices corresponding to the client's query. 
Moreover, the server needs to send the encrypted version of the entire database; otherwise, it would get an estimate about the number of records 
that match the query.

In Quantum Private Query (QPQ), the client issues queries to a database and obtains the values of the data bits corresponding to the queried indices such that the client can learn a small amount of extra information about the database bits that are not intended to know by her (known as database security), whereas the server can gain a small amount of information about the query indices of the client (known as user privacy) in a cheat-sensitive way. The functionality of this QPQ primitive can be explained as a probabilistic $n$-out-of-$N$ Oblivious Transfer (here we consider $n=1$) where the client has probabilistic knowledge about the other (the bits that are not intended to know by her) database bits. 

The first protocol in this domain had been proposed by Giovannetti et al.~\cite{GLM08}, followed by 
\cite{GLM10} and \cite{Ol11}. However, all these protocols used quantum memories and none of these are 
practically implementable at this point. For implementation purpose, Jakobi et al.~\cite{jakobi} presented an idea which was 
based on a Quantum Key Distribution (QKD) protocol~\cite{SARG}. This is the first QPQ protocol based on a QKD scheme. In 2012, Gao et al.~\cite{GLWC12} 
proposed a flexible generalization of~\cite{jakobi}. Later, Rao et al.~\cite{Rao} suggested two more efficient modifications of classical 
post-processing in the protocol of Jakobi et al. In 2013, Zhang et al.~\cite{Zhang} proposed a QPQ protocol based on the counterfactual QKD scheme~\cite{Noh}. Then, in 2014, Yang et al. came up with 
a flexible QPQ protocol~\cite{Yang} which was based on the B92 QKD scheme~\cite{b92}. This domain is still developing, as 
evident from the number of recent publications~\cite{Wei,Liu}. Some of these protocols exploit entangled states to generate a shared key 
between the server (Bob) and the client (Alice). In some other protocols, a single qubit is sent to the client. The qubit is prepared 
in certain states based on the value of the key and the client has to perform certain measurements on this encoded qubit to extract the key bit. 
Although these protocols differ in the process of key generation, the basic ideas are the same. The security of all these 
protocols is defined based on the following facts.
\begin{itemize}
	\item The server (Bob) knows the whole key which would be used for the encryption of the database.
	\item The client (Alice) knows a fraction of bits of the key.
	\item Bob does not get any information about the position of the bits which are known to Alice.
\end{itemize}
It is natural to consider that one of the legitimate parties may play the role of an adversary. Alice tries to extract more information about 
the raw key bits (which implies additional information about data bits), whereas Bob tries to know the position of the bits that are known to Alice. For this reason, QPQ can be viewed as a two-party mistrustful cryptographic primitive. 
Despite its cheat-sensitive property, the server Bob and the client Alice are allowed to violate user privacy and data security, respectively, with a negligible probability based on the security requirements. In practice, the exact primitive that one tries to achieve is as follows- 
\begin{itemize}
	\item Malicious Alice can only know a small amount of additional data bits than that is intended to know by her. Here the aim is to minimise Alice's of extra information about the database.
	\item Malicious Bob can only gain a small amount of information about the query indices of Alice. Here Alice tries to hides her query indices from Bob.
\end{itemize}

Very recently, Maitra et al.~\cite{MPR17} identified that the securities of all the existing protocols are based on the fact that the communicating parties rely on their devices, i.e., the source that supplies the qubits and the detectors that measure the qubits. Thus, similar to the QKD protocols, 
the trustworthiness of the devices are implicit in the security proofs of the QPQ protocols. However, in Device Independent (DI) scenario, these trustful assumptions over the devices are removed and the security is guaranteed even after removing these assumptions. But unlike QKD, it is hard to prove DI security in the case of QPQ because of its mistrustful property. 

To remove the trustful assumptions and enhance the overall security, recently a DI-QPQ protocol has been described in~\cite{MPR17} and it's finite sample analysis has been discussed in~\cite{BM18}. In~\cite{MPR17}, the authors introduced a testing phase at the server-side and proposed a semi-device independent version of the Yang et al.~\cite{Yang} QPQ scheme. 

In this QKD based QPQ scheme~\cite{Yang} (and also in the other QKD based QPQ schemes), the main idea of partial key generation at the client's side relies on the distinction between non-orthogonal states. For the QPQ scheme~\cite{Yang}, the server Bob and the client Alice share non maximally entangled states and Alice performs projective measurements at her side on some specified basis randomly to guess the raw key bits (chosen by Bob) with certainty. 

It is well-known that contrary to the non maximally entangled states, maximally entangled states are easy to prepare in practice and are also more robust in the case of DI certification. Moreover, it is also known that POVM measurement provides optimal distinction~\cite{iva87,PT98} between non-orthogonal quantum states. 

Keeping these in mind, here we propose a novel QPQ scheme using shared EPR pairs (between the server and the client) and POVM measurement (at the client's side to retrieve the maximum number of raw key bits with certainty). Our proposed scheme provides full DI certification exploiting self-testing of EPR pairs along with self-testing of POVM measurement (at the client's side) and projective measurement (at server's side). We further provide formal security proofs (considering all the strategies that preserve the correctness condition) and obtain an upper bound on the maximum cheating probabilities for both dishonest server and dishonest client.

\subsection{Relation between QPQ and Oblivious Transfer}

Oblivious Transfer (OT) is a well-studied cryptographic primitive which was first introduced informally by Wiesner~\cite{W83} and then subsequently formalized as 1 out of 2 OT in \cite{SGL85}. In $1$ out of $N$ oblivious transfer protocol, the server Bob has a database with $N$ entries, and the client Alice wants to know one of the entries with the intention that her choice would not be made public. To ensure privacy for both the server and the client, Bob shouldn't know anything about Alice's choice and Alice shouldn't know anything extra other than her choice. This scheme is also referred to as SPIR. However, there is a minimal difference between SPIR and OT. Generally, multiple databases are involved in SPIR schemes to achieve both low communication complexity and information-theoretic security. Thus, SPIR is considered as a distributed version of $1$ out of $N$ OT. 

On the other hand, private query protocols also offer a similar kind of functionalities. However, the security requirements of QPQ schemes are generally relaxed~\cite{GLM10} to the extent that Alice's privacy is cheat-sensitive (i.e., Bob may know about the choice of Alice after the transmission using some attacks). That means the QPQ security relies on the fact that if Bob tries to infer Alice's choice, she has a non-zero probability to discover it. One can verify that the more information Bob gets about Alice's choice, the higher is the probability that he will not pass Alice's security test. Furthermore, Alice may also obtain a few more entries other than her requirement (i.e., Alice may have probabilistic knowledge about the bits that are not intended to know by her) but the total number of entries that Alice can obtain is strictly bounded. To date, all the existing QPQ protocols are designed considering a single database. Therefore, QPQ is like a probabilistic $1$ out of $N$ OT as compared to SPIR.

It is already shown in~\cite{Lo97} that information-theoretically secure two-party computational schemes are impossible in the quantum scenario. This implies the impossibility of designing an OT scheme that satisfies both client's and server's security requirements. Fortunately, due to the aforementioned relaxed security requirements of QPQ~\cite{GLM10}, it is possible to design unconditionally secure QPQ schemes. The security of classical OT is generally based on computational complexity assumptions, whereas QPQ is information-theoretically secure. So, QPQ protocols can resist all attacks (even if the attacker uses quantum resources) whereas the classical or even quantum OT protocols may not be able to defend against such attacks.

\subsection{Comparison with the exact classical primitive}


It is well-known that in the classical setting, it is impossible to design an information-theoretically secure OT or SPIR schemes. However, to the best of our knowledge, it is not known whether we can design an information-theoretically secure \emph{classical private query} (CPQ) scheme. Here, we point out that, it is very easy to come up with a naive and inefficient information-theoretically secure classical private query scheme. A rough idea of the scheme is given below.

\begin{itemize}
	\item Suppose, the client Alice wants to know $I_1$ number of bits from the $N$ bit database $X$ but asks for $I_2$ positions (that include her $I_1$ positions) to the server Bob where $I_2$ is exponentially larger than $I_1$ but exponentially smaller than $N$.
	\item Bob then returns all the bits corresponding to these $I_2$ positions to Alice. This implies that Alice can't learn more than $I_2$ bits from the database which is very small compared to the size of the entire database.
	\item On the other hand, Bob can learn about the positions of Alice's query with probability $\frac{I_1}{I_2}$ which is also very small. 
\end{itemize}

One can easily check that although this naive classical solution is information-theoretically secure, it has the following disadvantages as compared to the existing quantum solutions. 

\begin{itemize}
	\item In the naive classical solution of the private query primitive, the server Bob leaks more data bits to the client Alice as compared to the existing quantum solutions. In the above-mentioned classical solution, Alice knows an exponential amount of additional data bits as compared to the size of her intended query index set. Whereas, in the quantum scenario, Alice knows a ver small amount of additional data bits compared to the size of her intended query index set.
	\item In the mentioned classical solution, Bob can guess the query indices of Alice with a more certain probability as compared to the existing quantum solutions. In the quantum scenario, Bob guesses each of the data bits as Alice's query with non-zero probability. Whereas, in the case of this mentioned classical solution, Bob can simply eliminate $(N-I_2)$ indices (exponential number of data bits as compared to the size of the query index set) that are not asked by Alice.  
\end{itemize}

The study of designing an efficient classical private query scheme is beyond the scope of this paper, and we leave it for our future work.


\subsection{Our Contribution}

In the current report, we address the problem of performing a private search on a classical database such that the user can retrieve an item from the database (along with some probabilistic knowledge about the other data bits) and the server can learn about the client's query in a cheat- sensitive way such that the data privacy and the user privacy are both preserved. 
As it is a distrustful cryptographic protocol, data privacy and, user security contradict each other. Moreover, because of the cheat sensitivity, if any of the parties try to violate data privacy or user security, it will be detected by the other party. The QPQ schemes are mainly different from the traditional QKD schemes in the following two aspects-

\begin{itemize}
	\item In QKD, the parties Alice and Bob both know all the bits of their shared raw key. However, in QPQ, only the server Bob knows all the bits of the shared raw key, and the client Alice knows only some bits of the shared raw key.
	\item In QKD, both the parties i.e., Bob and Alice trust each other and any third party will act as an adversary. Contrary to this, in QPQ, neither of the parties trust each other and any one of them may act as an adversary. 
\end{itemize} 

As QPQ is a distrustful cryptographic primitive, it is much harder to prove Device Independence (DI) in this setting. Keeping this in mind, in this proposed scheme, we try to maintain data privacy as well as user security (so that no significant information is leaked to any of the parties) and also try to maintain the cheat sensitive property (i.e., if any of the party tries to violate the security then this party will be caught by the other party).

Our main contribution in this paper is threefold which we enumerate below:

\begin{enumerate}

	\item We propose a novel QPQ scheme and remove the trustworthiness from the devices (source as well as measurement devices) using the self-testing of EPR pairs, self-testing of projective measurements (mentioned in \cite{kan17}) and self-testing of POVM measurements. Recently, Maitra et al.~\cite{MPR17} proposed a semi DI version of the QPQ scheme~\cite{Yang}. However, the QPQ scheme~\cite{Yang} uses non maximally entangled states which are difficult to prepare in practice and are also less robust in the case of DI certification as compared to the maximally entangled states. Keeping this in mind, here we propose a QPQ scheme using EPR pairs and a proper self-testing mechanism that guarantees full DI security of our protocol. To the best of our knowledge, this is the first time we provide such a full DI-QPQ scheme.

	\item We replace the usual projective measurement at client Alice's side with optimal POVM measurement so that (on average) Alice can obtain maximum raw key bits with certainty and (possibly) retrieve the maximum number of data bits in a single query. We also show that our proposed scheme provides (on average) the maximum number of raw key bits with certainty for Alice. 
	
	\item  Contrary to all the existing QPQ protocols, in the present effort, we provide a general security analysis (considering all the attacks that preserve the correctness condition) and provide an upper bound on the cheating probabilities (i.e., a lower bound on the amount of information leakage in terms of entropy) for both the parties (the server as well as the client).
\end{enumerate}
	
	\subsection{Notations and Definitions}
	Let us first list a few notations.
	\begin{itemize}
		\item $\K$: Initial number of states for our proposed scheme. Here, we assume that $\K$ is asymptotically large.
		\item $\mathbb{I}_k$: the Identity matrix of dimension $k$.
		\item $\mathcal{A} (\mathcal{A^*})$: honest (dishonest) client Alice.
		\item $\mathcal{B} (\mathcal{B^*})$: honest (dishonest) server Bob.
		\item $\mathcal{A}_i (\mathcal{A}_i^*)$: the $i$-th subsystem corresponding to honest (dishonest) Alice.
		\item $\mathcal{B}_i (\mathcal{B}_i^*)$: the $i$-th subsystem corresponding to honest (dishonest) Bob.
		\item $\ket{\phi}_{{\A}_i {\B}_i}$: the $i$-th copy of the shared state where 1st qubit corresponds to 
		Alice and 2nd qubit corresponds to Bob.
		\item $\rho_{{\A}_i{\B}_i}$: the density matrix representation for the $i$-th shared state.
		\item $\rho_{\mathcal{A}_i} (\rho_{\mathcal{B}_i})$: the reduced density matrix at Alice's (Bob's) side for $i$-th shared state.
		\item $X$: the $N$-bit database which corresponds to server Bob.
		\item $R (R_{\A})$: the entire raw key corresponding to Bob (Alice) of size $kN$ bits for some integer $k>1$.
		\item $F (F_{\A})$: the entire final key corresponding to Bob (Alice) of size $N$ bits.
		\item $R_i (R_{\A_i})$: the $i$-th raw key bit at Bob's (Alice's) side.
		\item $F_i (F_{\A_i})$: the $i$-th final key bit at Bob's (Alice's) side.
		\item $k$: Number of raw key bits XORed to generate every bit of the final key.
		\item $\mathcal{I}_l$: the index set of the elements which are queried by the client Alice.
		\item $l$: Size of Alice's query index set $\mathcal{I}_l$ i.e., $l=|\mathcal{I}_l|$.
		\item $a_i$: the classical bit announced by Bob for $i$-th shared state.
		\item $A (B)$: measurement outcome at Alice's (Bob's) side.
		\item $|0'\rangle = \cos \theta|0\rangle + \sin \theta |1\rangle$.
		\item $|1'\rangle = \sin \theta|0\rangle - \cos \theta |1\rangle$.
		\item $\in_{R}$: uniform random selection from a given set.
	\end{itemize}
	Next we present a few definitions that will be required for further discussions.
	\begin{itemize}
		\item {\bf Trace Distance:} The trace distance allows us to compare two probability distributions $\{ p_i \}$ and $\{ q_i \}$ 
		over the same index set which can be defined as
		\begin{eqnarray*}
			Dist(p_i, q_i)&=&\frac{1}{2} \sum_{i} |p_i - q_i|.
		\end{eqnarray*} 
	\end{itemize}
	\begin{itemize}
		\item In quantum paradigm, the trace distance is a measure of closeness of two quantum states $\rho$ and $\sigma$. The trace norm of an operator $M$ is defined as,
		\begin{eqnarray*}
			||M||_1&=& Tr{|M|},
		\end{eqnarray*}
		where $|M| = \sqrt{M^{\dagger}M}$.
		The trace distance between quantum states $\rho$ and $\sigma$ is given by,
		\begin{eqnarray*}
			Dist(\rho, \sigma)&=& Tr|\rho - \sigma|  \\
			&=& ||\rho - \sigma ||_1,
		\end{eqnarray*}
		where $|A| = \sqrt{A^{\dagger} A}$ is the positive square root of $\sqrt{A^{\dagger} A}$. 
	\end{itemize}
	\begin{itemize}
		\item  {\bf Fidelity:} Like trace distance, fidelity is an alternative measure of closeness. In terms of fidelity, the similarity 
		between the two probability distributions $\{ p_i \}$ and $\{ q_i \}$ can be defined as,  
		\begin{eqnarray*}
			F(p_i, q_i)&=& \left(\sum_{i} \sqrt{p_i q_i}\right)^2. 
		\end{eqnarray*}
	\end{itemize}
	\begin{itemize}
		\item The fidelity of two quantum states $\rho$ and $\sigma$ is defined as
		\begin{eqnarray*}
			F(\rho, \sigma)&=& \left[Tr (\sqrt{\rho^{1/2} \sigma \rho^{1/2}})\right]^2. 
		\end{eqnarray*}
	\end{itemize}
	\begin{itemize}
		\item In case of pure states, the fidelity is a squared overlap of the states $|\psi\rangle$ and $|\phi\rangle$, i.e.,
		\begin{eqnarray*}
			F(\rho, \sigma)&=& |\langle \psi | \phi \rangle |^2,
		\end{eqnarray*}
		where $\rho = |\psi\rangle\langle\psi|$ and $\sigma = |\phi\rangle\langle\phi|$ are corresponding density matrix representation of the pure states $|\psi\rangle$ and $|\phi\rangle$ respectively.
	\end{itemize}
	\begin{itemize}
		\item The two measures of closeness of quantum states, trace distance and fidelity, are related by the following inequality \cite{FG99},
		$$1 - \sqrt{F(\rho, \sigma)} \leq \frac{1}{2} Tr|\rho - \sigma| \leq \sqrt{1 - F(\rho, \sigma)}.$$ 
	\end{itemize}
	\begin{itemize}
		\item Trace distance has a relation with the distinguishability of two quantum states. Suppose, one referee prepares two quantum states $\rho$ and $\sigma$ for 
		another party (say Alice) to distinguish. The referee prepares each of the states with probability $\frac{1}{2}$. Let $p_{\text{correct}}$ denotes the optimal 
		guessing probability for Alice and it is related to trace distance by the following expression, 
		$$p_{correct} = \frac{1}{2} \left(1 + \frac{1}{2} Tr|\rho - \sigma |\right).$$
		It implies that trace distance is linearly dependent to the maximum success probability in distinguishing two quantum 
		states $\rho$ and $\sigma$. For further details one may refer to \cite{Wilde17}.
	\end{itemize}
	\begin{itemize}
		\item {\bf Conditional Minimum Entropy:} Let $\rho = \rho_{AB}$ be the density matrix representation of a bipartite quantum 
		state. Then the conditional minimum entropy of subsystem $A$ conditioned on subsystem $B$ is defined by (\cite{KRS09}) 
		\begin{eqnarray*}
			H_{min}(A|B)_{\rho} &=& -\inf_{\sigma_B} D_{\infty}(\rho_{AB}||\mathbb{I}_A \otimes \sigma_B),
		\end{eqnarray*}
		where $\mathbb{I}_A$ denotes the identity matrix of the dimension of system $A$ and the infimum ranges over all normalized density operators 
		$\sigma_B$ on subsystem $B$ and also for any two density operators $T,T'$ we define, $$D_{\infty}(T||T') = \inf\{ \lambda \in \mathbb{R}: T \leq 2^{\lambda} T' \}.$$ 
	\end{itemize}
	\begin{itemize}
		\item Let $\rho_{XB}$ be a bipartite quantum state where the $X$ subsystem is classical. For the given state $\rho_{XB}$ if $p_{\text{guess}}(X|B)_{\rho_{XB}}$ denotes the maximum probability of guessing $X$ given the subsystem $B$, then from \cite{KRS09} we have, 
		\begin{equation}
		\label{pguess}
		p_{\text{guess}}(X|B)_{\rho_{XB}} = 2^{-H_{min}(X|B)_{\rho}}.
		\end{equation}
	\end{itemize}

	
	\subsubsection{Adversarial Model}
	
	As Quantum Private Query (QPQ) is a distrustful cryptographic primitive, here each of the parties have different security goals. The security requirement of the entire protocol is termed as \textit{Protocol Correctness} whereas the security of the server (Bob) is termed as \textit{Data Privacy} and the security requirement for the client (Alice) is termed as \textit{User privacy}. Formally, these terms are defined below.
	
	\begin{definition}{Protocol Correctness:}
		\label{prot_corr}
		
		If the user (i.e., the client) Alice and the database owner (i.e., the server) Bob both are honest, then after the \kep, the probability that Alice can correctly retrieve the expected number of data bits in a single database query is very high. This implies that in case of honest implementation of the protocol, if $X$ denotes the actual number of data bits known by Alice and $E[X]$ denotes the expected number of data bits that are supposed to be known by Alice then, after the \kep,
		
		\begin{equation}
		\noindent
		\Pr(|X-E[X]| \leq \delta_t \wedge \text{the protocol does not abort}) \geq P_c
		\end{equation} 
		
		where $\delta_t$ denotes the amount of deviation allowed by Bob and $P_c$ denotes the probability with which the value of $X$ lies within the interval $[E[X]-\delta_t,E[X]+\delta_t]$ (ideally, the value of $P_c$ should be high). 
		
	\end{definition}

	\begin{definition}{Protocol Robustness:}
		\label{prot_rob}
		
		If the user (i.e., the client) Alice and the database owner (i.e., the server) Bob both are honest, then after the \kep~of our proposed scheme, the probability that Alice will know none of the final key bits (as well as the database bits) and the protocol has to be restarted is very low. More formally,
		
		\begin{equation}
		\noindent
		\Pr(\text{the protocol aborts in honest scenario}) \leq P_a
		\end{equation}  
		
		where $P_a$ denotes the probability that Alice knows none of the final key bits and aborts the protocol (ideally, the value of $P_a$ should be small).
		
	\end{definition}

		
		\begin{definition}{Data Privacy:}
			\label{data_priv}
			
			A QPQ protocol satisfies the data privacy property if either the protocol aborts with high probability in the asymptotic limit, or in a single database query, dishonest Alice's strategy ($\A^*$) can correctly extract (on average) at most $\tau$ fraction of bits of the $N$-bit database $X$ where $\tau (0 < \tau < 1)$ is very small compared to the size of the entire database i.e., $N$. This implies that if $D_{\A^*}$ denotes the number of data bits that dishonest Alice can extract (on average) in a particular query then,
			\begin{equation}
			\noindent
			E_R(D_{\A^*}) \leq \tau N
			\end{equation}
			where the expectation is taken all over the random coins $R$ that are used in the protocol.

			The data privacy against dishonest Alice can also be defined in terms of the success probability in guessing more than the expected number of data bits. In this notion, after the \kep, either the protocol aborts with high probability in the asymptotic limit, or the probability that dishonest Alice ($\A^*$) correctly retrieves more than the expected number of data bits and the protocol does not abort is very low. This implies that for dishonest Alice ($\A^*$), if $X$ denotes the actual number of data bits known by Alice and $E[X]$ (i.e, $E_R(D_{\A^*})$) denotes the expected number of data bits that are actually supposed to know by Alice then after the \kep, 
			
			\begin{equation}
			\noindent
			\Pr(|X-E[X]| > \delta_t \wedge \text{the protocol doesn't abort}) \leq P_d
			\end{equation} 
			
			where $\delta_t$ denotes the amount of deviation allowed by Bob and $P_d$ denotes the probability with which the value of $X$ lies outside the interval $[E[X]-\delta_t,E[X]+\delta_t]$ such that the protocol does not abort (ideally, the value of $P_d$ should be very small). 
			
		\end{definition}

		\begin{definition}{User Privacy:}
			\label{user_priv}
			
			Let $\ind_l = \{i_1, \ldots, i_l \}$ denotes the indices of the data bits that Alice wants to know from the database by performing $l$ many queries. Then for a QPQ protocol, after $l$ many queries, either the protocol aborts with high probability in the asymptotic limit, or the dishonest Bob's strategy ($\B^*$) can correctly guess (on average) at most $\delta$ fraction of indices from the index set $\ind_l$ where $\delta$ ($0 < \delta <1$) is very small compared to the size of the index set i.e., $l$. This implies that after $l$ many queries to the database by Alice, if $\ind_{\B^*}$ denotes the number of correctly guessed indices by dishonest Bob then, 
			
			\begin{equation}
			\noindent
			E_{R'}(\ind_{\B^*}) \leq \delta l
			\end{equation}
			
			where the expectation is taken all over the random coins $R'$ that are used in the protocol.
			
			The user privacy against dishonest Bob can also be defined in terms of the success probability in guessing a query index correctly from Alice's query index set. In this notion, either the protocol aborts with high probability in the asymptotic limit, or the probability that dishonest Bob ($\B^*$) can correctly guess an index from Alice's query index set ($\ind_l$) and the protocol does not abort is very low. This implies that after $l$ many queries to the database by Alice, if dishonest Bob guesses an index $i$ from the database and the protocol does not abort then the probability that $i$ is in $\ind_l$ is very low i.e.,
			
			\begin{equation}
			\noindent
			\Pr(\text{Bob guesses } i \in \ind_l \wedge \text{the protocol does not abort}) \leq P_u
			\end{equation} 
			
			where $P_u$ denotes the probability that $i$ is in $\ind_l$ and the protocol does not abort (ideally, the value of $P_u$ should be very small).
		\end{definition}


	\subsection{Assumptions for Our Device Independent Proposal}
	
	In this section we mention the list of assumptions that are required for the security of our proposed QPQ scheme. Those assumptions are summarized as follows.
	
	\begin{enumerate}
		\item Devices follow the laws of quantum mechanics i.e., the quantum states and the measurement operators involved in this scheme lead to the observed outcomes via the Born rule.
			\item Like the recent DI proposal for oblivious transfer from the bounded-quantum-storage-model and computational assumptions in~\cite{BY21}, here also we assume that the state generation device and the measurement devices (both at honest and dishonest party's end) are described by a tensor product of Hilbert spaces, one for each device. That means for this proposal, we assume that the devices follow the {\it i.i.d.} assumptions such that each use of a device is independent of the previous use and they behave the same in all trials. This also implies that the statistics of all the rounds are independent and identically distributed (i.i.d.) and the devices are memoryless. We also assume that the honest party chooses the inputs randomly and independently for each rounds.
			
			\emph{Note :} As QPQ is a distrustful primitive, to detect the fraudulent behavior (if any) of the dishonest party, the {\it i.i.d.} assumption on the inputs chosen by the honest party seems justified here. It is also possible to consider more general scenarios for this distrustful primitive without imposing the {\it i.i.d} assumption on the devices, but these are outside the scope of this work.
			
			\item The honest party can interact with the unknown devices at his end only by querying the devices with the inputs and getting the corresponding outputs whereas the dishonest party can manipulate all the devices before the start of the
			protocol. However, we assume that after the protocol starts, the dishonest party can no longer change this behavior - s/he cannot manipulate any devices held by the honest party, and also cannot ``open up" any devices s/he possesses at her/his end (i.e. the dishonest party is also restricted to only
			supplying the inputs and getting the corresponding outputs from those devices after the start of the protocol). We also
			assume that the dishonest party processes their data in an {\it i.i.d.} fashion.
			
		\item Generally, in the Device Independent (DI) scenario, it is assumed that Alice's and Bob's laboratories are perfectly secured, i.e., there is no communication between the laboratories. As QPQ is a distrustful scheme, here we assume that each party's aim is not only to retrieve as much additional information as possible from the other party but also to leak as little additional information as possible from his side. For this reason, while testing the cheating of a dishonest party in a particular testing phase, the other party must act honestly in that test to detect the fraudulent behavior (if any) of the dishonest party. If both the parties act deceitfully in any testing phase, then none of them can detect the cheating of the other party. So, one party must act honestly in every testing phase. 
		
		Here we assume that in the local tests, the party who acts honestly in a particular testing phase chooses the input bits randomly for the devices at his end (on behalf of the referee), and self-tests the devices. So, for local tests, there is no communication between the laboratories. 
		
		Similarly, whenever they perform distributed tests (i.e., the tests performed by both of them with the shared states), we assume that the honest party chooses the input bits for both the parties on behalf of the referee. Then the parties measure their qubits accordingly, and the dishonest party announces the output bits for the samples chosen for the test.
		
		This implies that for the local tests involved in this scheme, there is no communication between the laboratories, whereas, for the distributed tests, we allow communication regarding the input bits and the output bits from the honest party's end and the dishonest party's end respectively.
		
		We also assume that the honest party can somehow ``shield" his devices such that no information (regarding the inputs and the outputs) is leaked from his laboratory until he chooses to announce something.
		
		\emph{Note :} Here, one may think that in case of distributed test, the dishonest party may not measure his qubits according to the values of the input bits chosen by the honest party. In that case, how the honest party can detect this dishonest behaviour in the corresponding testing phase is clearly mentioned later in the analysis of {\it device independent security}. 
		\item The inputs for self-tests are chosen freely and independently of all the other systems involved in the protocol i.e., the device used to generate input bits for one party does not have any correlations (classical or quantum) with the particles or the source or the laboratory of the other party.
	\end{enumerate}

\section{Proposal for a full DI-QPQ scheme}

The QPQ protocols are composed of several phases. Depending on the functionality, we have divided the entire protocol into five phases. 
The first phase is termed ``entanglement distribution phase". In this phase, a third party (need not be a trusted one and may collude with the dishonest party) distributes several copies of entangled states between the server (Bob) and the client (Alice). The next phase is called ``source device verification phase". 
In this phase, the server and the client self-test their shared entangled states using CHSH game. The third phase is termed as ``DI testing for Bob's measurement device". In this phase, Bob self-tests his measurement device (in some specific measurement basis that will be used for the QPQ protocol). 

In QPQ, before the protocol, the server Bob decides how much information the client Alice can retrieve from the database in a single query. For this reason, Bob chooses a parameter $\theta$ and performs measurements on his qubits (of the shared entangled states) in this $\theta$ rotated basis (during the protocol) to restrict Alice's information about the database~\footnote{Once chosen, this value of $\theta$ remains fixed for the entire QPQ protocol}. As Alice and Bob get the measurement devices from an untrusted third party, (in device-independent setting) they need to check the devices before proceeding with the protocol. Here we assume that dishonest Bob's aim is not only to know Alice's query indices but also to leak as little additional information about the database as possible. For this reason, in ``DI testing phase for Bob's measurement device", only Bob will act as a referee and choose input bits for both parties. They first perform some measurements assuming the devices as unknown boxes and then after getting the outcome, they conclude about their functionality. After measurement, if the probability of winning the specified game is equal to some predefined value, then they can conclude that Bob's measurement devices are noiseless for those specified bases. 

The next phase of this protocol is termed ``DI testing for Alice's POVM elements". In the phase, Alice first performs specific measurements assuming the POVM devices as unknown boxes and then concludes about their functionality based on the outcome i.e., in this phase Alice checks the functionality of her POVM device. If the POVM device works as expected, then Alice and Bob generate key bits in the next phase for the remaining instances which is termed as ``Key Establishment Phase". After this phase, Bob has a secret key such that Alice knows some of those bits and Bob doesn't know the indices of the bits known by Alice.

In the last phase, i.e., in ``private query phase", Bob encrypts the database using the key generated at his side and sends the encrypted database to Alice. Alice then decrypts the intended data bits using 
the known key bits at her side.\\

Now we describe different steps of our proposal. Note that we haven't considered the channel noise here. So, we assume here that all the operations are perfect. 
	
	\begin{enumerate}
		\item \textbf{Entanglement Distribution Phase:}
		\begin{enumerate}
			\item A third party distributes $\K$ copies (where we assume that $\K$ is asymptotically large) of EPR pairs $|\phi\rangle_{\A\B}$ between Alice and Bob such that Alice (Bob) receives $\A$ ($\B$) 
			subsystem of $|\phi\rangle_{\A\B}$. 
		\end{enumerate}

		\item \textbf{Source Device Verification Phase:}
		
		The source device verification phase is composed of two subphases. In the first subphase, Bob acts as a referee, chooses random samples (for testing phase), receives the corresponding qubits from Alice, generates random input bits for those instances and performs a localCHSH test to certify the states. Similarly, in the second subphase, Alice acts as a referee and does the same that Bob does in the previous phase. In each phase, after receiving the inputs, Alice's and Bob's device measure the states and return output bits $(c_i, b_i)$. The detail description of different subphases is as follows.

		\IncMargin{-1em}
		\RestyleAlgo{boxruled}
		\LinesNumbered
		\begin{algorithm}[htbp]
			
			\begin{itemize}
				\item For each $i \in \mathcal{S}$, $\mathcal{P}$ does the following- 
				\begin{enumerate}
					\item For the inputs $s_i = 0$ and $s_i = 1$, $\mathcal{P}$'s device performs a measurement in the first qubit of the $i$-th state and outputs $c_i = 0$ or $c_i = 1$. 
					\item For the inputs $r_i = 0$ and $r_i = 1$, $\mathcal{P}$'s device performs a measurement in the second qubit of the $i$-th state and outputs $b_i = 0$ or $b_i = 1$. 
				\end{enumerate} 
				\item From the inputs $s_i,r_i$ and corresponding outputs $c_i,b_i$, $\mathcal{P}$ estimates the following quantity, 
				\hspace{-0.8in}	
				\begin{align*}
				\mathcal{C} &= \frac{1}{|\mathcal{S}|} \sum_{i \in \mathcal{S}}\mathcal{C}_i\\
				\end{align*}
				where the parameter $\mathcal{C}_i$ is defined as follows ,
				\begin{align*}
				\mathcal{C}_i := 
				\begin{cases}
				1  \quad \text{If }s_ir_i = c_i \oplus b_i\\
				0  \quad \text{otherwise}.
				\end{cases}
				\end{align*}
				\item If $\mathcal{C} = \cos^2{\frac{\pi}{8}}$\footnote{In the case of honest implementation, this exact desired value can be obtained for asymptotically large number of samples. However, in practice, with finite number of samples, it is nearly always impossible to exactly match with the desired value of the estimated statistic. Hence, a small deviation from the desired value is allowed in practice. A discussion regarding the variation of the deviation range with the sample size is mentioned in Appendix A. However, how the existing security definitions will vary with the noise parameter, is out of the scope of this present work and we will try to explore this issue in our future works. } then the protocol continues.
				\item Otherwise, the protocol aborts.
			\end{itemize}
			
			\caption{LoaclCHSHtest($\mathcal{S}, \mathcal{P}$)}
			\label{chshtest} 
		\end{algorithm}
		\DecMargin{-1em}

		
		\begin{enumerate}
			\item Bob randomly chooses $\frac{\gamma_1 \K}{2}$ (how Bob and Alice choose the specific value of $\gamma_1$ from the set $(0,1)$ is mentioned in Appendix A) instances from the shared $\K$ instances, declares the instances publicly and constructs a set $\Gamma^{\B}_{\CHSH}$ with these instances.
			\item Alice sends her qubits for all the instances in $\Gamma^{\B}_{\CHSH}$ to Bob.
			\item For these instances in $\Gamma^{\B}_{\CHSH}$, Bob plays the role of the referee as well as the two players and plays local CHSH game.
			\item  For every $i$-th sample in $\Gamma^{\B}_{\CHSH}$, Bob generates random bits $r_i \in_R \{0,1\}$ and $s_i \in_R \{0,1\}$ as the inputs of his two measurement devices (these devices act as the devices of two different parties).  
			\item Bob performs LocalCHSHtest($\Gamma^{\B}_{\CHSH}$, Bob), mentioned in algorithm \ref{chshtest} (which is nothing but the local version of actual CHSH game) for the set $\Gamma^{\B}_{\CHSH}$.
			\item If Bob passes the LocalCHSHtest($\Gamma^{\B}_{\CHSH}$, Bob) game then they proceed further, otherwise they abort.
			\item From the rest $\left(\K-\frac{\gamma_1 \K}{2}\right)$ instances, Alice randomly chooses $\frac{\gamma_1 \K}{2}$ (how Bob and Alice choose the specific value of $\gamma_1$ from the set $(0,1)$ is mentioned in Appendix A) instances, declares the instances publicly and constructs a set $\Gamma^{\A}_{\CHSH}$ with these instances.
			\item Bob sends her qubits for all the instances in $\Gamma^{\A}_{\CHSH}$ to Alice.
			\item For these instances in $\Gamma^{\A}_{\CHSH}$, Alice plays the role of the referee as well as the two players and plays local CHSH game.
			\item  For every $i$-th sample in $\Gamma^{\A}_{\CHSH}$, Alice generates random bits $r_i \in_R \{0,1\}$ and $s_i \in_R \{0,1\}$ as the inputs of her two measurement devices (these devices act as the devices of two different parties).  
			\item Alice performs LocalCHSHtest($\Gamma^{\A}_{\CHSH}$, Alice), mentioned in algorithm \ref{chshtest} (which is nothing but the local version of actual CHSH game) for the set $\Gamma^{\A}_{\CHSH}$.
			\item If Alice passes the LocalCHSHtest($\Gamma^{\A}_{\CHSH}$, Alice) game then they proceed to the next part of the protocol where Bob self-tests his observables, otherwise they abort. 
		\end{enumerate}	  
		

		\item \textbf{DI Testing for Bob's Measurement Device:}

		\IncMargin{-1em}
		\RestyleAlgo{boxruled}
		\LinesNumbered
		\begin{algorithm}[htbp]
			
			\begin{itemize}
				\item For each $i \in \mathcal{S}$, Alice and Bob does the following- 
				\begin{enumerate}
					\item Bob generates a random bit $s_i \in_R \{0,1\}$ as an input of Alice's device and declares the input publicly. 
					\item For the inputs $s_i = 0$ and $s_i = 1$, Alice's device performs a measurement in her part of the $i$-th copy of the shared states and outputs $c_i = 0$ or $c_i = 1$. 
					\item Bob already generates the input bits $r_i=0$ or $r_i=1$ randomly for his measurement device in the $i$-th instance and obtains the outcome $b_i=0$ or $b_i=1$.  
					\item Alice and Bob declare their inputs $s_i,r_i$ and corresponding outputs $c_i,b_i$.
				\end{enumerate} 
				\item From the declared outcomes Alice and Bob estimate the following quantity, 
				\hspace{-0.8in}	
				\begin{align*}
				\beta &= \frac{1}{4} \sum_{s,r,c,b \in \{0,1\}}(-1)^{d_{srcb}}\alpha^{1\oplus s} \langle \phi_{\A\B}| A^s_c \otimes B^r_b |\phi_{\A\B}\rangle\\
				\end{align*}
				where $\alpha = \frac{(\cos{\theta}+\sin{\theta})}{|(\cos{\theta}-\sin{\theta})|}$ and $d_{srcb}$ is defined as follows ,
				\begin{align*}
				\label{cuv}
				d_{srcb} := 
				\begin{cases}
				0  \quad \text{If }sr = c \oplus b\\
				1  \quad \text{otherwise}.
				\end{cases}
				\end{align*}
				\item If $\beta = \frac{1}{\sqrt{2}|(\cos{\theta}-\sin{\theta})|}$\footnote{In the case of honest implementation, this exact desired value can be obtained for asymptotically large number of samples. However, in practice, with finite number of samples, it is nearly always impossible to exactly match with the desired value of the estimated statistic. Hence, a small deviation from the desired value is allowed in practice. A discussion regarding the variation of the deviation range with the sample size is mentioned in Appendix A. However, how the existing security definitions will vary with the noise parameter, is out of the scope of this present work and we will try to explore this issue in our future works. } then the protocol continues.
				\item Otherwise, the protocol aborts.
			\end{itemize}
			
			\caption{OBStest($\mathcal{S}$)}
			\label{obstest} 
		\end{algorithm}
		\DecMargin{-1em}
		
		\begin{enumerate}		        
			\item Let us assume that in the previous phase (i.e., in \evp), Bob and Alice select total $|\Gamma_{\CHSH}|$ samples where $\Gamma_{\CHSH}=\Gamma^{\A}_{\CHSH} \cup \Gamma^{\B}_{\CHSH}$. In this phase, they consider the rest $(\K-|\Gamma_{\CHSH}|)$ instances and for $1\leq i \leq (\K-|\Gamma_{\CHSH}|)$, Bob does the following-
			\begin{itemize}
				\item Bob first generates a random bit $r_i$ (i.e., bit $r_i \in_{R} \{0,1\}$) for the $i$-th instance (basically these random bits are the raw key bits at Bob's side i.e., $R_i = r_i$).
				\item If $r_i = 0$, Bob's device applies measurement operator $\{B_0^0, B_1^0 \}$, and generates the output $b_i=0$ and $b_i=1$ respectively.
				\item If $r_i = 1$, Bob's device applies measurement operator $ \{B_0^1, B_1^1\}$, and generates the output $b_i=0$ and $b_i=1$ respectively.
				\item Bob declares $a_i = 0$ whenever his device outputs $b_i=0$ (i.e., the device applies measurement operator $B^0_0$ or $B^1_0$ for the $i$-th instance).
				\item Bob declares $a_i = 1$ whenever his device outputs $b_i=1$ (i.e., the device applies measurement operator $B^0_1$ or $B^1_1$ for the $i$-th instance). 
			\end{itemize}
			
			\item From these $(\K-|\Gamma_{\CHSH}|)$ instances, Bob randomly chooses $\frac{\gamma_2(\K-|\Gamma_{\CHSH}|)}{2}$ (how Bob and Alice choose the specific value of $\gamma_2$ from the set $(0,1)$ is mentioned in Appendix A) instances, declares the instances publicly and constructs a set $\Gamma^{\B}_{\obs}$ with these instances.
			\item Alice then randomly chooses $\frac{\gamma_2(\K-|\Gamma_{\CHSH}|)}{2}$ (how Bob and Alice choose the specific value of $\gamma_2$ from the set $(0,1)$ is mentioned in Appendix A) instances from the rest $(\K-|\Gamma_{\CHSH}|-|\frac{\gamma_2(\K-|\Gamma_{\CHSH}|)}{2}|)$ instances, declares the instances publicly and make a set $\Gamma^{\A}_{\obs}$ with these instances. 
			\item Alice and Bob construct a set $\Gamma_{\obs}$ with all their chosen instances i.e., $\Gamma_{\obs} = \Gamma^{\A}_{\obs} \cup \Gamma^{\B}_{\obs}$.
			\item Alice and Bob perform OBStest($\Gamma_{\obs}$), mentioned in algorithm \ref{obstest}, for the set $\Gamma_{\obs}$.

		\end{enumerate}

		\item \textbf{DI Testing for Alice's POVM Elements:}
		\begin{enumerate}
			\item After the DI testing phase for Bob's measurement device, Alice and Bob proceed to this phase with the rest $(\K-|\Gamma_{\CHSH}|-|\Gamma_{\obs}|)$ shared states. Let us denote this set as $\Gamma_{\text{POVM}}$.
			\item Alice randomly chooses $\gamma_3 |\Gamma_{\text{POVM}}|$ (how Alice chooses $\gamma_3$ is mentioned in Appendix A) samples from the rest shared $|\Gamma_{\text{POVM}}|$ states. We call this set as $\Gamma^{\text{test}}_{\text{POVM}}$. Alice performs KEYgen($\Gamma^{\text{test}}_{\text{POVM}}$), mentioned in algorithm \ref{povmtest}, for the set $\Gamma^{\text{test}}_{\text{POVM}}$.
			\item Alice then performs POVMtest($\Gamma^{\text{test}}_{\text{POVM}}$), mentioned in algorithm \ref{POVMtest}, for the set $\Gamma^{\text{test}}_{\text{POVM}}$.  
		\end{enumerate}
		

		\IncMargin{-1.0em}
		\RestyleAlgo{boxruled}
		\LinesNumbered
		\begin{algorithm}[htbp]
			
			\begin{itemize}
				\item For each $i \in \mathcal{S}$, Alice does the following-
				\label{stp} 
				\begin{enumerate}
					\item If Bob declared $a_i = 0$, Alice measures her qubit of the $i$-th shared state using the measurement device $M^0=\{M^0_0, M^0_1, M^0_2\}$.
					\item If Bob declared $a_i = 1$, Alice measures her qubit of the $i$-th shared state using the measurement device $M^1=\{M^1_0, M^1_1, M^1_2\}$.
				\end{enumerate}
			\end{itemize}
			\caption{KEYgen($\mathcal{S}$)}
			\label{povmtest}
		\end{algorithm}
		\DecMargin{-1.0em}
		

		\IncMargin{-1em}
		\RestyleAlgo{boxruled}
		\LinesNumbered
		\begin{algorithm}[htbp]
			\begin{itemize}
				\item Alice considers all those instances of the set $\mathcal{S}$ where Bob declared $a_i=0$ and creates a set $\mathcal{S}^0$ with those instances.
				\item Similarly, with the rest of the instances (i.e., the instances where Bob declared $a_i=1$), Alice creates a set $\mathcal{S}^1$.
				\item Let us assume that $y$ denotes the value of $a_i$ and for the set $\mathcal{S}^y$, the states at Alice's side are either $\rho^y_x$ or $\rho^y_{x \oplus 1}$ (for input $x \in_R \{0,1\}$ at Bob's side). 
				\item For each set $\mathcal{S}^y$, Alice calculates the value of the parameter \begin{align*}
				\Omega^{y}&=\sum_{b,x \in\{0,1\}} (-1)^{b\oplus x} \tr[M^{y}_b\rho^{y}_x]
				\end{align*} 
				where $M^y_b$ is the measurement outcome at Alice's side in KEYgen().
				\item If for every $\mathcal{S}^y$ ($y \in \{0,1\}$), $$\Omega^{y} = \frac{2\sin^2{\theta}}{(1+\cos{\theta})}$$\footnote{In the case of honest implementation, this exact desired value can be obtained for asymptotically large number of samples. However, in practice, with finite number of samples, it is nearly always impossible to exactly match with the desired value of the estimated statistic. Hence, a small deviation from the desired value is allowed in practice. A discussion regarding the variation of the deviation range with the sample size is mentioned in Appendix A. However, how the existing security definitions will vary with the noise parameter, is out of the scope of this present work and we will try to explore this issue in our future works. } then the protocol continues. 
				\item Otherwise, the protocol aborts.
			\end{itemize}
			\caption{POVMtest($\mathcal{S}$)}
			\label{POVMtest} 
		\end{algorithm}
		\DecMargin{-1.0em}

		\item \textbf{Key Establishment Phase:}
		\begin{enumerate}
			\item After the DI testing phase for POVM elements, Alice proceeds to this phase with the rest $(|\Gamma_{\text{POVM}}|-\gamma_3 |\Gamma_{\text{POVM}}|)$ shared states. Let us denote this set as $\Gamma_{\text{Key}}$.
			\item For the shared states of the set $\Gamma_{\text{Key}}$, Alice performs 
			KEYgen$(\Gamma_{\text{Key}})$.
			\item After KEYgen$(\Gamma_{\text{Key}})$,
			\begin{itemize}
				\item If Alice gets $M^0_0 (M^0_1)$ for $a_i=0$, she concludes that the original raw key bit for $i$-th instance is $0 (1)$. Whenever Alice gets $M^0_2$, she ignores that outcome.
				\item Similarly, if Alice obtains $M^1_0 (M^1_1)$ for $a_i=1$, she concludes that the original raw key bit for $i$-th instance is $0(1)$. Whenever Alice gets $M^1_2$, she ignores that outcome.
			\end{itemize}
			\item After these key generation, Alice and Bob proceed to private query phase with this $|\Gamma_{\text{Key}}|$ shared states. Note that $|\Gamma_{\text{Key}}| = kN$ for some positive integer $k > 1$ where $N$ is the number of bits in the database and $k$ is exponentially smaller than $N$.
			\item Alice and Bob use the raw key bits obtained from these $kN$ many states for the next phase.
		\end{enumerate}

		\item \textbf{Private Query Phase:}
		\begin{enumerate}
			\item Alice and Bob now share a raw key of length $kN$ bits where Bob knows every bit value and Alice knows partially (and Bob doesn't know the indices of the bits known by Alice).   
			\item Bob randomly announces a permutation which reorder the $kN$ bit string. After the announcement, they both apply the permutation on their raw key bits.
			\item Bob cuts the raw key into $N$ sub strings of length $k$ and tells each bit position to Alice. The bits of every sub string are added bit wise by Alice and Bob to form the 
			final key of length $N$. At the end, if Alice does not know any bit of the final key $F$ (which actually corresponds to Bob), then the protocol has to be executed again.
			\item Now suppose that Alice knows only the $i$-th bit $F_i$ of Bob's final key $F$ and wants to know the $j$-th bit $m_j$ of the database, then she announces a permutation $P_A$ such that after applying the permutation, the $i$-th bit of the final key goes to $j$-th position. 
			Consequently, Bob applies this permutation $P_A$ on the final key $F$ and use it to encrypt the database using one time pad. As, $m_j$ will be encrypted by $F_i$, Alice can correctly 
			recover the intended bit after receiving the encrypted database. 
			\item If Alice knows only one final key bit and wants to know the information about $l$ database bits, she has to announce the permutation for $l$ many times to retrieve the intended bits.
			\item If Alice knows more than one bit of the final key and wants to know more than one bit of the database in a single trial then Alice announces the permutation in such a way that her known key bits encrypt the intended database bits which she wants to retrieve. Thus, Alice can retrieve more than one intended data bits in a single trial. 
		\end{enumerate}
	\end{enumerate} 

	\onecolumngrid

	\begin{tcolorbox}
		
		{\bf Our QPQ Proposal (In Case of Honest Implementation)}
		
		\begin{itemize}
			\item The server Bob and the client Alice share $\K$ EPR pairs among themselves such that the first qubit of every shared EPR state corresponds to Alice and the second qubit corresponds to Bob.
			\item For each of these $\K$ shared EPR pairs, Bob and Alice generate raw key bits in the following way-
			\begin{itemize}
				\item Bob randomly chooses the value of the $i$-th raw key bit $r_i$ (i.e., $r_i \in_R \{0,1\}$).
				\item If $r_i=0$, Bob measures his qubit of the $i$-th shared state in $\{\ket{0},\ket{1}\}$ basis, otherwise (i.e., for $r_i=1$) he measures in $\{\ket{0'},\ket{1'}\}$ basis where $\ket{0'}=(\cos{\theta}\ket{0}+\sin{\theta}\ket{1})$ and $\ket{1'}=(\sin{\theta}\ket{0}-\cos{\theta}\ket{1})$ (here Bob chooses the value of $\theta$ according to the relation as mentioned in equation \ref{relation4}).
				\item Bob declares a classical bit $a_i=0 (a_i=1)$ whenever the measurement outcome at his side for the $i$-th instance is either $\ket{0} (\ket{1})$ or $\ket{0'} (\ket{1'})$.
				\item Whenever Bob declared $a_i = 0$, Alice measures her qubit of the $i$-th EPR pair using the POVM $M^0=\{M^0_0, M^0_1, M^0_2\}$ where
				\begin{eqnarray*}
					M^0_0 &\equiv& \frac {(\sin{\theta}\ket{0} - \cos{\theta}\ket{1}) (\sin{\theta}\bra{0} - \cos{\theta}\bra{1})} {1 + \cos{\theta}}\\
					M^0_1 &\equiv& \frac{1}{1 + \cos{\theta}} \ket{1} \bra{1}\\
					M^0_2 &\equiv& I - M^0_0 - M^0_1
				\end{eqnarray*}
				\item Similarly, whenever Bob declared $a_i = 1$, Alice measures her qubit of the $i$-th EPR pair using the POVM $M^1=\{M^1_0, M^1_1, M^1_2\}$ where
				\begin{eqnarray*}
					M^1_0 &\equiv& 
					\frac{(\cos{\theta}\ket{0} + \sin{\theta}\ket{1})(\cos{\theta}\bra{0} + \sin{\theta}\bra{1})}{1 + \cos{\theta}}\\
					M^1_1 &\equiv& \frac{1}{1 + \cos{\theta}} \ket{0} \bra{0}\\
					M^1_2 &\equiv& I - M^1_0 - M^1_1 
				\end{eqnarray*}
				\item If Alice gets $M^0_0 (M^0_1)$ for $a_i=0$, she concludes that the original raw key bit for $i$-th instance is $0 (1)$. Whenever Alice gets $M^0_2$, her measurement outcome remains uncertain.
				\item Similarly, if Alice obtains $M^1_0 (M^1_1)$ for $a_i=1$, she concludes that the original raw key bit for $i$-th instance is $0(1)$. Whenever Alice gets $M^1_2$, her measurement outcome remains uncertain.
				\item After this raw key generation phase, Bob and Alice perform some postprocessing (permutation and XOR) on their raw key bits to generate the final key of the size equals to the size of the database (Initially, Bob chooses the value of $\theta$ and the number of raw key bits to generate every bit of the final key according to the relation mentioned in equation \ref{relation4}). After post processing, Bob knows all the bits of the final key whereas Alice knows only some of the bits.
				\item Bob then encrypts the entire database using his final key and sends the encrypted database to Alice. 
				\item Finally, Alice decrypts the intended data bits using her partial knowledge about the final key.
			\end{itemize} 
		\end{itemize}
		
	\end{tcolorbox}

		\begin{figure}[ht!]
			\includegraphics[scale=0.4]{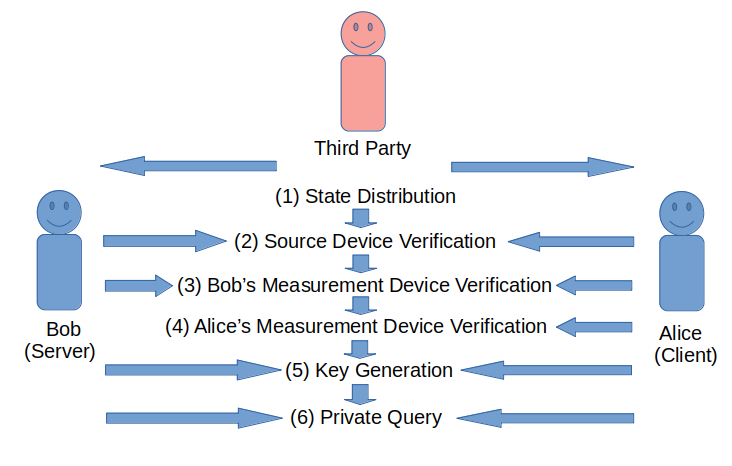}
			\caption{Schematic diagram of our proposed DI-QPQ scheme.
				In step (1), third party distributes EPR pairs between Alice and Bob such that the first qubit of each state corresponds to Alice and the second qubit corresponds to Bob. In step (2), Alice and Bob randomly choose some instances from their shared EPR pairs to certify the states. In step (3), Alice and Bob randomly choose some shared instances (from the rest) to self-test Bob's projective measurement device. In step (4), Alice chooses some shared instances randomly (from the rest) to self-test her POVM device.\newline In step (5), Bob and Alice perform projective and POVM measurement respectively on the rest shared instances to generate shared raw key such that Bob knows all the bits and Alice knows some of the bits.\newline In step (6), Alice performs private query to retrieve the intended data bits.}
			\label{rawsamp}
		\end{figure}

	\twocolumngrid

	\section{Analysis of the protocol} 
	
	In this section we discuss the functionality of our proposed scheme. At first, we discuss the correctness of our protocol in subsection A. Next, we estimate (in subsection B) the values of different parameters involved in our scheme for security purpose. Atlast, we discuss the security related issues of our proposed scheme in subsection C. 
	
	Note that here we present all our analyses considering the asymptotic scenario. In reality, the values of different parameters (derived here) may deviate from their derived value depending on the chosen sample size.

\subsection{Correctness of the Protocol} 
		First we prove the correctness of the protocol.
		\begin{theorem}
			\label{correct}
			In honest Alice and honest Bob scenario, at the end of  $\kep$, Alice can correctly guess (on average) only $(1 - \cos{\theta})kN$ many bits of the entire raw key $R$. 
		\end{theorem}
		\begin{proof} 
			
			After the key establishment phase, Bob and Alice share $kN$ raw key bits. These raw key bits were generated from $kN$ copies of maximally entangled states of the form
			\begin{eqnarray*}
				&&\frac{1}{\sqrt{2}}(\ket{0}_{\A}\ket{0}_{\B}+\ket{1}_{\A}\ket{1}_{\B})\\
				&=&\frac{1}{\sqrt{2}}(\ket{0'}_{\A}\ket{0'}_{\B}+\ket{1'}_{\A}\ket{1'}_{\B}),
			\end{eqnarray*}
			where, $\ket{0'}= (\cos{\theta}\ket{0}+\sin{\theta}\ket{1})$ and $\ket{1'}= (\sin{\theta}\ket{0}-\cos{\theta}\ket{1})$. Here $\theta$ may vary from $0$ to $\frac{\pi}{2}$.
			
			The generation of such $kN$ raw key bits can be redefined as follows.
			
			Bob prepares a random bit stream $R = r_1 \ldots r_{kN}$ of length $kN$. If $r_i = 0$, Bob measures his qubits in $\{\ket{0}, \ket{1}\}$ basis. 
			Whereas, if $r_i=1$, Bob measures his qubit in $\{\ket{0'}, \ket{1'}\}$ basis. After each measurement Bob announces a bit $a_i\in\{0,1\}$. 
			If he gets $\ket{0}$ or $\ket{0'}$, he announces $a_i=0$. If he gets $\ket{1}$ or $\ket{1'}$, he announces $a_i=1$. Now, Alice's job is to guess the value of each $r_i$.
			
			Thus, whenever Bob declares $a_i = 0$, Alice can understand that Bob gets either $\ket{0}$ or $\ket{0'}$ and the shared qubit of her side also collapses to $\ket{0}$ or $\ket{0'}$ 
			respectively. However, to obtain the value of the raw key bit, Alice has to distinguish these two states with certainty. As, $\ket{0}$ and $\ket{0'}$ are 
			non-orthogonal states (when $\theta\neq\frac{\pi}{2}$), Alice cannot distinguish these two states with certainty for all the instances. 
			
			According to the strategy mentioned in the protocol, whenever Bob declares $a_i = 0$, Alice chooses the POVM $\{M^0_0, M^0_1, M^0_2\}$. After measurement, if Alice receives the outcome 
			$M^0_0$, she concludes that Bob's measurement outcome was $\ket{0}$. In such case, Alice concludes that $r_i=0$. If Alice receives the outcome $M^0_1$, 
			she concludes that Bob's measurement outcome was $\ket{0'}$. In such a case, Alice concludes that $r_i=1$. However, if the measurement outcome is $M^0_2$, then Alice remains 
			uncertain about the value of the raw key bit. Alice follows the similar methodology for $a_i=1$.
			
			Now, we calculate the success probability of Alice to guess each $r_i$ correctly. Let us assume that $\Pr(M^{a_i}_j|\ket{\phi^{a_i}_i})$ denotes the corresponding success probability of 
			getting the result $M^{a_i}_j$ when the given state is $\ket{\phi^{a_i}_i}$ i.e., $$\Pr(M^{a_i}_j|\ket{\phi^{a_i}_i}) = \bra{\phi^{a_i}_i}M^{a_i}_j\ket{\phi^{a_i}_i}.$$ We now calculate the corresponding success 
			probabilities of getting different results for the states $\ket{0}$ and $\ket{0'}$.
			
			For $\ket{0}$, the success probabilities will be
			\begin{eqnarray*}
				\Pr(M^0_0|\ket{0}) &=&  \bra{0}M^0_0\ket{0}\\
				& =& (1 - \cos{\theta})\\
				\Pr(M^0_1|\ket{0}) &=&  \bra{0}M^0_1\ket{0} \\
				&=& 0\\
				\Pr(M^0_2|\ket{0}) &=&  \bra{0}M^0_2\ket{0} \\
				&=& \cos{\theta}
			\end{eqnarray*}
			Similarly, for the state $\ket{0'}$, the success probabilities will be
			\begin{eqnarray*}
				\Pr(M^0_0|\ket{0'}) &=&  \bra{0'}M^0_0\ket{0'}\\
				&=& 0\\
				\Pr(M^0_1|\ket{0'}) &=&  \bra{0'}M^0_1\ket{0'}\\
				& =& (1 - \cos{\theta})\\
				\Pr(M^0_2|\ket{0'}) &=&  \bra{0'}M^0_2\ket{0'}\\
				&=& \cos{\theta}
			\end{eqnarray*}
			Similarly, whenever Bob declares $a_i = 1$, Alice chooses the POVM $\{M^1_0, M^1_1, M^1_2\}$. In a similar way, we can calculate the success probability here. 
			We formalize all the conditional probabilities in the following table.
			
			{\small
				\begin{center}
					\begin{tabular}{|c|c|c|c|c|c|c|c|}
						\hline
						& \multicolumn{4}{|c|}{Cond. Probability of Alice}\\
						\cline{2-5}
						a &\theadfont\diagbox[width=5 em]{Bob}{Alice}& A=$M^0_0 / M^1_0$ & A=$M^0_1 / M^1_1$ & A=$M^0_2 / M^1_2$ \\
						\hline
						0 & $ B=\ket{0} $ & $1-\cos{\theta}$ & $0$ & $\cos{\theta}$\\
						\hline
						0 & $B=\ket{0'}$ & $0$ & $1-\cos{\theta}$ &$\cos{\theta}$\\
						\hline
						1 & $B=\ket{1}$ & $1-\cos{\theta}$ & $0$ &$\cos{\theta}$\\
						\hline
						1 & $B=\ket{1'}$ & $0$ & $1-\cos{\theta}$ &$\cos{\theta}$\\
						\hline
					\end{tabular}
				\end{center}}
				According to the protocol, if $a_i=0$ and 
				Alice gets $M^0_0(M^0_1)$, she outputs $r_{\A_i}=0 (1)$. When $a_i = 1$ and she gets $M^1_0(M^1_1)$, she outputs $r_{\A_i}=0 (1)$. 
				Thus, the success probability of Alice to guess the $i$-th raw key bit $r_i$ of Bob can be written as 
				\begin{eqnarray}
				\label{sucprob}
				& &\Pr(r_{\A_i} = r_i) \nonumber\\
				&=&\Pr(r_{\A_i}=0,r_i=0) + \Pr(r_{\A_i}=1,r_i=1) \nonumber\\
				&=& (1 - \cos{\theta}).\nonumber
				\end{eqnarray}
				So, according to the proposed scheme, the overall success probability of Alice in guessing a raw key bit is equal to $(1 - \cos{\theta})$. This implies that at the end of the key establishment phase, Alice can guess (on average) around $(1 -\cos{\theta})kN$ many raw key bits with certainty. 
			\end{proof}

\subsection{Parameter estimation for private query phase}
	
	Here, we estimate the values of different parameters such that the protocol preserves both user privacy and data privacy. In this scheme, after the \kep, Bob has $kN$ many raw key bits such that Bob knows all these bits and Alice knows some of these bits. In \pqp, both Alice and Bob cut their raw keys in some particular positions to prepare $N$ sub strings of length $k$ such that $k=\frac{|\Gamma_{\text{Key}}|}{N}$ where $|\Gamma_{\text{Key}}|$ denotes the total number of raw key bits at the \pqp and $N$ denotes the number of database bits. Alice and Bob then perform bit wise XOR among the bits of each sub string to get the $N$ bit final key $F$. Here, $r_i (1 \leq i \leq kN)$ denotes the i-th raw key of Bob and $f_i (1 \leq i \leq N)$ denotes the i-th final key of Bob. Based on the procedure mentioned in \pqp for generating final key bits, the relation between $r_i$ and $f_i$ can be written as,
	$$f_i = \oplus_{j=(i-1)k + 1}^{ik} r_j ~~~~(1 \leq i \leq N)$$ 
	Where $\oplus$ denotes addition modulo 2.
	
	It will be clearer by a toy example. Consider $N = 10$ and $k = 2$. Let us assume that the raw key at Bob's side is,
	$$01~ 10~ 01~ 00~ 10~ 01~ 01~ 11~ 00~ 11$$
	and after the \kep, the raw key at Alice's side is,
	$$?1~ ??~ 0?~ ??~ ??~ 01~ ?1~ ??~ 0?~ ?1$$
	i.e., Alice knows the values of 2nd, 5th, 11th, 12th, 14th, 17th and 20th key bits of the original raw key ($?$ stands for inconclusive key bit i.e., the positions where Alice can't guess the key bits with certainty).
	
	Now, after the modulo operation on the raw key, the final key at Bob's side will be,
	$$1~ 1~ 1~ 0~ 1~ 1~ 1~ 0~ 0~ 0$$
	and the final key at Alice's side will be,
	$$?~ ?~ ?~ ?~ ?~ 1~ ?~ ?~ ?~ ?$$ 
	Thus, the number of known key bits by Alice is reduced from 7 to 1. 
	The significance of such modulo operation is to enhance the security of the protocol. This is similar to the privacy amplification in a QKD protocol.\\
	
	\underline{\emph{Estimation of the parameter $\theta$ for security purpose :}}\\
	
			
			In this proposal, the client Alice can guess a raw key bit with probability around $\left(1 - \cos{\theta}\right)$ and both the server Bob and the client Alice share $kN$ many raw key bits. So, if $n_r$ denotes the number of raw key bits that Alice knows after the \kep~of our scheme then we can write that the expected value of $n_r$ will be,
			
			\begin{align}
			E[n_r] &= \left(1 - \cos{\theta}\right)kN
			\end{align}

			Alice and Bob then XOR $k$ number of raw key bits to generate every bit of the final key. So, to conclusively guess a final key bit, Alice has to guess all the $k$ raw key bits correctly corresponding to that final key bit. This implies that Alice can correctly guess a final key bit with probability around $P_f = \left(1 - \cos{\theta}\right)^k$. 
			
			Let us assume that the variable $n_f$ denotes the number of final key bits known by Alice. It is easy to check that $n_f$ is a binomially distributed random variable with a total of $N$ number of final key bits and the probability of getting one final key bit successfully is $P_f = \left(1 - \cos{\theta}\right)^k$.
			
			So, the expected number of final key bits that Alice knows after the \kep~is,
			
			
			
			\begin{align}
			\label{final}
			E[n_f] &= P_f N \approx \left(1 - \cos{\theta}\right)^k N   
			\end{align}
			
			

			In our DI proposal, dishonest Alice has to measure correctly (using the specified POVM) to successfully pass the DI testing phase. Moreover, it is known that $(1 - \cos{\theta})$ is the optimal probability in distinguishing two non orthogonal states~\cite{iva87}. This implies that in our proposed scheme, whenever the protocol does not abort, dishonest Alice can guess the $i$-th raw key bit $R_i$ with probability atmost $\left(1 - \cos{\theta}\right)$ i.e.,
			
			\begin{align}
			\label{rawupper}
			\Pr[R_{\A^*_i} = R_i] &\leq \left(1 - \cos{\theta}\right)
			\end{align} 
			
			where $\A^*_i$ denotes dishonest Alice's subsystem corresponding to the $i$-th shared state. 
			
			As after Bob's measurement, the states at Alice's side are independent and the measurement devices at dishonest Alice's side are also independent and memoryless, the guessing probability of dishonest Alice for the $i$-th final key bit $F_i$ will be upper bounded by $\left(1 - \cos{\theta}\right)^k$ i.e., 
			
			\begin{align}
			\label{finalprob}
			\Pr[F_{\A^*_i} = F_i] = P_f &\leq \left(1 - \cos{\theta}\right)^k
			\end{align} 
			
			From equation \ref{final} and equation \ref{finalprob}, we can conclude that the expected number of final key bits that dishonest Alice can guess correctly whenever Bob does not abort the protocol will be upper bounded by,
			
			\begin{align}
			\label{totalprob}
			E[F_{\A^*}] & \leq \left(1 - \cos{\theta}\right)^k N 
			\end{align}
		
	According to our proposal, the size of the final key is equal to the size of the database and the database is encrypted by performing bit wise XOR with the final key. So, the correct guessing of a final key bit implies the correct guessing of the corresponding database bit. This implies that whenever the protocol does not abort, the expected number of data bits that dishonest Alice can guess correctly in a single query is also upper bounded by $\left(1 - \cos{\theta}\right)^k N$ i.e., 

\begin{equation}
	\noindent
	\label{finaldata}
	E[D_{\A^*}] \leq \left(1 - \cos{\theta}\right)^k N
\end{equation}

As in our scheme, the server Bob wants the client Alice to know atleast a single final key bit (so that the protocol does not abort) and always less than two final key bits, the expected number of final key bits known by Alice must satisfy the condition,

			\begin{align*}
			1 &\leq E[n_f] < 2 
			\end{align*}
			
			This implies that,
			
			\begin{align}
			\label{relation4}
			1 &\leq \left(1 - \cos{\theta}\right)^k N  < 2 \nonumber\\
			\frac{1}{N} &\leq \left(1 - \cos{\theta}\right)^k  < \frac{2}{N}
			\end{align} 
			
	
	From these results, one can conclude the following,
	
	\begin{corollary}
		\label{lower_pc}
		
		If the server Bob wants the client Alice to know at least one final key bit but less than two final key bits, then Bob needs to choose the parameter $k$ and the value of $\theta$ such that,
		
		\begin{equation}
		\noindent
		\frac{1}{N} \leq \left(1 - \cos{\theta}\right)^k < \frac{2}{N} \nonumber
		\end{equation}
	\end{corollary}


Now from the correctness condition, we derive an upper bound on the value of the parameter $P_a$ (mentioned in definition \ref{prot_rob}) and a lower bound on the value of the parameter $P_c$ (mentioned in definition \ref{prot_corr}) for our scheme.\\


	
\underline{\emph{ Estimation of the security parameters $P_a$ and $P_c$ :}}\\

Atfirst, we calculate the probability that the protocol does not abort in honest scenario. Then from equation \ref{relation4}, using the derived bound on the value of $\left(1 - \cos{\theta}\right)^k$, we can obtain a lower bound on the value of $P_c$ using the Chernoff-Hoeffding inequality~\cite{WH} (here we estimate the value of $P_c$ using Chernoff-Hoeffding inequality because we assume that dishonest Alice also measures in i.i.d fashion). 

	
	In our scheme, the probability that Alice can successfully guess a final key bit is (approximately) equal to $\left(1 - \cos{\theta}\right)^k$.
	
	So, the probabilty that Alice can not guess a final key bit is equal to $\left[1 - \left(1 - \cos{\theta}\right)^k\right]$.
	
	This implies that the probability that Alice knows none of the $N$ final key bits is equal to 
	
	\begin{align}
	\left[1 - \left(1 - \cos{\theta}\right)^k\right]^N &\approx e^{-\left(1 - \cos{\theta}\right)^k N} 
	\end{align}
	
	i.e., for our proposed scheme, we get the following bound on the value of the parameter $P_a$.
	
	\begin{align}
	\label{pa1}
	P_a &\leq  e^{-\left(1 - \cos{\theta}\right)^k N}
	\end{align}
	
	From equation \ref{relation4}, we get that $\frac{1}{N} \leq \left(1 - \cos{\theta}\right)^k < \frac{2}{N}$. If we consider that Bob chooses the values of $\theta$ and $k$ such that $\left(1 - \cos{\theta}\right)^k = \frac{1}{N}$ then replacing this value in equation \ref{pa1}, we can get,
	
	\begin{equation}
	\noindent
	\label{pa2}
	\boxed{P_a \leq e^{-1}}
	\end{equation} 
	
	This implies that the value of $P_a$ is small for our proposed scheme. Similarly, the probability that the protocol does not abort in honest scenario (i.e., Alice knows atleast one final key bit) is equal to
	
	\begin{align}
	\label{protsucc}
	\Pr(\text{the protocol doesn't abort}) &\geq \left[1 - e^{-1}\right]
	\end{align}

	So, for our proposed scheme, the probability that the protocol does not abort is high. Now, we recall the Chernoff-Hoeffding~\cite{WH} inequality.
	
	\begin{proposition}
		\label{propchernoff}
		(Chernoff-Hoeffding Inequality) Let $X=\frac{1}{m}\sum_{1 \leq i \leq m}{X_i}$ be the average of $m$ independent random variables $X_1, X_2, \cdots, X_m$ with values $(0,1)$, and let $\mathbb{E}[X]=\frac{1}{m}\sum_{1\leq i \leq m}{\mathbb{E}[X_i]}$ be the expected value of $X$. Then for any $\delta_{CH}>0$, we have $\Pr\left[|X-\mathbb{E}[X]| \geq \delta_{CH} \right] \leq \exp(-2\delta_{CH}^2 m)$.
	\end{proposition} 
	
	After the \kep~of our scheme, we consider $X_i=1$ whenever Alice knows the value of the $i$-th final key bit (i.e., for all the raw key bits corresponding to the $i$-th final key bit, Alice gets a conclusive POVM outcome i.e., either $M^0_0$ or $M^0_1$ for $a_i=0$ and either $M^1_0$ or $M^1_1$ for $a_i=1$) and $X_i=0$ otherwise. As there are total $N$ number of final key bits, we consider the value of the random variable $X$ as $X = \sum_{i=1}^{N} X_i$. 
	
	From the correctness of our scheme, we can say that whenever the protocol does not abort, the expected number of final key bits that Alice should know after the \kep~is $E[X]=\left(1 - \cos{\theta}\right)^k N$ and there are total $m=N$ number of final key bits. Now, we want that the value of $X$ lie within the error margin $\delta_{CH}=\epsilon \left(1 - \cos{\theta}\right)^k N$ (where the value of $\epsilon$ is very small and in practice, this value depends on the number of samples chosen for a particular testing phase. One may refer to Appendix A to check how the specific value of $\epsilon$ can be chosen for a given number of samples) from the expected value. Here, we can calculate the corresponding probability using the Chernoff-Hoeffding inequality as the final key bits at dishonest Alice's side are all independent (because the states collapsed at Alice's side are all independent and the measurement devices at dishonest Alice's end are independent and memoryless). For this proposed scheme, the value of the random variable $X$ and also the expected value $E[X]$ is calculated considering the scenario that the protocol does not abort. So, from the expression of Chernoff-Hoeffding bound in proposition \ref{propchernoff}, we can write that,
	
	\begin{align}
	\noindent
	\label{expx}
	\nonumber
	& \Pr\left[|X-\mathbb{E}[X]| < \delta_{CH}  \wedge \text{protocol doesn't abort} \right] \\
	& \geq 1-\exp(-2\delta_{CH}^2 m) 
	\end{align}
	
	After the \kep, Alice and Bob share $N$ number of final key bits and we want that the number of final key bits known by Alice lie within the interval $[p-\epsilon p, p+\epsilon p]$ where $p=\left(1 - \cos{\theta}\right)^k N$ and the deviation allowed here is $\delta_{CH} = \epsilon \left(1 - \cos{\theta}\right)^k N$. From the expression \ref{expx}, replacing the value of $\delta_{CH}$ and $m$, we get that, 
	
	\begin{equation}
	\noindent
	\label{expxf}
	\boxed{
		\begin{aligned}
		& \Pr\left[|X-\mathbb{E}[X]| < \delta_{CH} \wedge \text{protocol doesn't abort} \right] \\
		& \geq 1-\exp(-2\delta_{CH}^2 N) \\
		& \text{where}~~\delta_{CH} = \epsilon \left(1 - \cos{\theta}\right)^k N 
		\end{aligned}}
	\end{equation}
	
	We already get the following bound on the value of $\left(1 - \cos{\theta}\right)^k$ from equation \ref{relation4} for our proposed scheme.
	
	\begin{equation*}
	\noindent
	\frac{1}{N} \leq \left(1 - \cos{\theta}\right)^k < \frac{2}{N}
	\end{equation*}
	
	So, if we consider that Bob chooses the values of $\theta$ and $k$ such that $\left(1 - \cos{\theta}\right)^k=\frac{1}{N}$ then replacing this value in equation \ref{expxf}, we can get,

	\begin{equation}
	\noindent
	\label{expxf2}
	\boxed{
		\begin{aligned}
		& \Pr\left[|X-\mathbb{E}[X]| < \epsilon \wedge \text{protocol doesn't abort} \right] \\ 
	& \geq 1-\exp(-2\epsilon^2 N)
	\end{aligned}}
	\end{equation}
	
	According to our proposal, the correct guessing of a final key bit implies the correct guessing of the corresponding data bit. So, from definition \ref{prot_corr}, we can say that in our proposed scheme, for honest Alice and honest Bob, the value of the parameter $P_c$ is lower bounded by,
	
	\begin{equation}
	\noindent
	\label{boundpc}
	\boxed{P_c \geq [1-\exp(-2\epsilon^2 N)]}
	\end{equation}
	
	In practice, this probability is high as the value of $N$ is very large. This implies that in honest scenario of our proposed scheme, the probability that Alice knows the expected number of final key bits (with atmost $\epsilon$ deviation from the expected number) and the protocol does not abort is high. 
	
	
	It is already mentioned that Bob chooses the values of $\theta$ and $k$ such that Alice knows atleast one and less than two final key bits. From the relation mentioned in equation \ref{relation4}, we can get the following bound on the value of $\delta_{CH}$. 
	
	\begin{equation}
	\noindent
	\label{deltach}
	\epsilon  \leq \delta_{CH} < 2 \epsilon 
	\end{equation}
	
	From this relation, one can easily argue that the upper bound on the value of $\epsilon$ can be derived from the inequality $2 \epsilon \leq 1$ and the corresponding upper bound will be $\epsilon \leq \frac{1}{2}$.
	


	To evaluate the performance, here we consider our scheme as $1$ out of $2$ probabilistic oblivious transfer (i.e., $N=2$ and $k=1$). From equation \ref{totalprob} and equation \ref{relation4}, we can argue that for the $1$ out of $2$ probabilistic oblivious transfer variant of our scheme, if Bob chooses the value of $\theta$ such that $(1-\cos{\theta}) = \frac{1}{2}$ (i.e., the minimum value for $N=2$ and $k=1$), then the expected number of final key bits (or data bits) that Alice can retrieve in a single query is $\left(\frac{1}{2} \times 2\right) = 1$. From equation \ref{protsucc}, we can say that if we consider the $1$ out of $2$ probabilistic oblivious transfer variant of our scheme then in honest scenario,
	
	\begin{equation}
	\label{12protsucc}
	\Pr(\text{protocol doesn't abort}) \geq (1 - e^{-1}) \approx 0.632
	\end{equation}
	
	Similarly, from equation \ref{boundpc}, one can conclude that for the variant $1$ out of $2$ probabilistic oblivious transfer, if we consider $\epsilon = \frac{1}{2}$ then the probability that Alice gets the expected number of final key bits and the protocol doesn't abort is lower bounded by 
	
	\begin{equation}
	\label{12pcvalue}
	P_c \geq (1 - e^{-1}) \approx 0.632
	\end{equation}
	


	\subsection{Security of the Protocol}
	
	In this section, we discuss the security related issues of our proposed scheme. Based on the results in Corollary \ref{cor3}, Theorem \ref{cor5}, Theorem \ref{thm:povmtest0} and Theorem \ref{thm:povmtest1}, we conclude about the DI security of our proposed scheme. All these results guarantee that either the protocol aborts with high probability in the asymptotic limit or the devices involved in the scheme achieve the intended values of the parameters $\mathcal{C}$, $\beta$, $\Omega^0$ and $\Omega^1$. Later on, we move towards deriving upper bounds on the information gained by dishonest Alice and dishonest Bob. In Lemma \ref{min_ent2}, we show that dishonest Alice cannot get (on average) more than $(1- \cos \theta)$ fraction of bits of the entire raw key. Lemma \ref{min_ent3} together with corollary \ref{upper_delta} show that dishonest Bob can guess only $\frac{l}{N}$ fraction of indices from Alice's query index set.

\subsubsection{\textbf{Device independent security}} 

~\\

In our proposed scheme, the device independent (DI) testing has been done in three phases. The first two DI testing are done in {\em source device verification phase} and {\em DI testing for Bob's measurement device}. The third DI testing occurs in {\em DI testing phase for Alice's POVM elements}. 


In source device verification phase, at first {\em Local CHSH game} has been performed by each of Alice and Bob independently (as mentioned in LocalCHSHtest) at their end for some randomly chosen samples. In this phase, both Alice and Bob test individually whether the states provided by the third party are EPR pairs. Bob and Alice choose the samples randomly for which they want to perform LocalCHSHtest and share this information publicly to get the corresponding qubits from the other party and also to identify all the samples for which they perform LocalCHSHtest. 


As QPQ is a distrustful scheme, both the parties may not behave honestly in every phase of the protocol. For this reason, here we assume that the party who acts honestly for a particular phase, will take the responsibilities of the referee as well as the two parties in the CHSH game to ensure the random and independent choice of inputs for the devices involved in the LocalCHSHtest at his end. This guarantees that in LocalCHSHtest, the inputs to the devices are random and independent.

So, from the rigidity of CHSH game~\cite[Lemma 4.2]{ruv13}, one can conclude the following.

\begin{corollary}[DI testing of shared states]
	\label{cor3}
	In the LocalCHSHtest of \evp, either the devices achieve $\mathcal{C} = \cos^2 \frac{\pi}{8}$ for both Alice and Bob (i.e., the states provided by the third party are EPR pairs), or the protocol aborts with high probability in the asymptotic limit. 
\end{corollary} 

In the next phase, Bob checks the functionality of his measurement device. At first Bob chooses the inputs randomly for his device and measures his particles accordingly. After that, first Bob and then Alice choose samples independently from the rest shared states and discuss publicly about those chosen instances. Then for the chosen samples, Bob generates the input bits randomly for Alice and announce the bits publicly so that Alice can measure her particles according to these bit values. 

Here we assume that Bob will act honestly in this phase to check the functionality of his devices because from the result in Lemma \ref{min_ent3}, it is clear that if dishonest Bob wants to guess Alice's query indices with more certain probability then he should allow Alice to know more data bits in a single query which violates our assumption 4 that none of the parties reveal additional information from his side to get more information from the other party.

After the measurements, Bob and Alice discuss all their inputs and outputs publicly and calculate the value of the parameter $\beta$ as mentioned in the OBStest. From this result, one can conclude the following.

\begin{theorem}[DI testing of Bob's measurement devices]
	\label{cor5}
	In OBStest, either Bob's measurement devices achieve $\beta = \frac{1}{\sqrt{2}|(\cos{\theta}-\sin{\theta})|}$ (i.e., his devices measure correctly in $\{\ket{0}, \ket{1}\}$ and $\{\ket{0'}, \ket{1'}\}$ basis where $\ket{0'} = (\cos \theta\ket{0} + \sin \theta\ket{1})$, $\ket{1'} = (\sin \theta\ket{0} - \cos \theta\ket{1})$), or the protocol aborts with high probability in the asymptotic limit.
\end{theorem} 

The detail proof of this theorem is given in Appendix B which follows exactly the same approach that is mentioned in~\cite{kan17} for certifying non-maximally incompatible observables. 

This implies that the LocalCHSHtest certifies the states provided by the third party and OBStest certifies the projective measurement device (for the specific measurement bases used in OBStest) of Bob. As Bob declares $a_i$ values for all the shared instances before OBStest and Alice randomly chooses some of those instances for OBStest, the successful completion of OBStest also implies that for all the remaining instances (i.e., for the instances which are not chosen for OBStest), Alice's state must be either $\ket{0}\bra{0}$ or $\ket{0'}\bra{0'}$ whenever Bob declares $a_i=0$ and must be either $\ket{1}\bra{1}$ or $\ket{1'}\bra{1'}$ whenever Bob declares $a_i=1$.

The third DI testing is done in {\it DI testing phase for Alice's POVM elements}. According to the protocol, Alice and Bob lead to this phase whenever they have successfully passed the first two DI testing phases. So, Alice and Bob are in this phase implies that both Bob's projective measurement device and their shared states are noiseless. Now, this testing phase basically guarantees the functionality of Alice's POVM device. Note that in this phase, Bob need not test his measurement device again. During OBStest, his devices are tested already. However, Alice has to shift to a new measurement device for better conclusiveness. Device independent security demands that Alice's new device should be tested further for certification. In this phase, Alice measures the chosen states using the device $M^0 =\{M^0_0, M^0_1, M^0_2\}$ or $M^1 =\{M^1_0, M^1_1, M^1_2\}$ based on the declared $a_i$ values for each of those instances. From the measurement outcomes, she computes the quantity $\Omega^0$ and $\Omega^1$ (as defined in $POVMtest()$) and checks whether each of these values equal to $\frac{2\sin^2{\theta}}{(1+\cos{\theta})}$. Theorem \ref{thm:povmtest0} (Theorem \ref{thm:povmtest1}) shows that, for the instances where $a_i=0$ ($a_i=1$), if Alice observes that $\Omega^0 = \frac{2\sin^2{\theta}}{(1+\cos{\theta})}$  $\left(\Omega^1 = \frac{2\sin^2{\theta}}{(1+\cos{\theta})}\right)$ then it guarantees that the measurement devices are the desired POVM $\{D^0_0,D^0_1,D^0_2\}$ i.e., $M^0=D^0$ ($\{D^1_0,D^1_1,D^1_2\}$ i.e., $M^1=D^1$). 

\begin{theorem}[DI testing of Alice's measurement device $M_0$]
	\label{thm:povmtest0}
	In POVMtest, for the instances where Bob declares $a_i=0$, either the protocol aborts with high probability in the asymptotic limit or Alice's measurement devices achieve $\Omega^0 = \frac{2 \sin^2 \theta}{1+\cos \theta}$ i.e., the devices are of the following form (up to a local unitary),
	
	\begin{align}
	M^0_0 & = \frac{1}{(1+\cos \theta)} (|1'\rangle \langle1'|)\\
	M^0_1 & = \frac{1}{(1+\cos \theta)} (|1\rangle \langle1|)\\
	M^0_2 & = \id - M^0_0- M^0_1,
	\end{align}
	
	where $|1'\rangle = \sin \theta |0\rangle - \cos \theta |1\rangle$.
\end{theorem}

\begin{theorem}[DI testing of Alice's measurement device $M_1$]
	\label{thm:povmtest1}
	In POVMtest, for the instances where Bob declares $a_i=1$, either the protocol aborts with high probability in the asymptotic limit or Alice's measurement devices achieve $\Omega^1 = \frac{2 \sin^2 \theta}{1+\cos \theta}$, i.e., the devices are of the following form (up to a local unitary),
	
	\begin{align}
	M^1_0 & = \frac{1}{(1+\cos \theta)} (|0'\rangle \langle0'|)\\
	M^1_1 & = \frac{1}{(1+\cos \theta)} (|0\rangle \langle0|)\\
	M^1_2 & = \id - M^1_0- M^1_1,
	\end{align}
	
	where $|0'\rangle = \cos \theta |0\rangle + \sin \theta |1\rangle$.
\end{theorem}

The proofs of these two theorems are deferred to the subsection entitled {\it DI testing of POVM elements} of the Appendix C. In the proof, we restate the functionality of the POVM devices in the form of a two party game (namely POVMgame), consider a general form for the single qubit three outcome POVM $\{M^0_0,M^0_1, M^0_2\}$ ($\{M^1_0,M^1_1, M^1_2\}$) and show that if the input states are chosen randomly between $\ket{0}\bra{0} (\ket{1}\bra{1})$ and $\ket{0'}\bra{0'} (\ket{1'}\bra{1'})$ and if $\Omega^0 = \frac{2 \sin^2 \theta}{1+\cos \theta}$ ($\Omega^1 = \frac{2 \sin^2 \theta}{1+\cos \theta}$) then $M^0_0 = D^0_0$ ($M^1_0 = D^1_0$), $M^0_1 = D^0_1$ ($M^1_1 = D^1_1$), $M^0_2 = D^0_2$ ($M^1_2 = D^1_2$).\\



\emph{Note:}~Here, we claim that if Alice and Bob successfully pass the LocalCHSH test, the OBStest and the POVMtest mentioned in our DI proposal, then in the actual QPQ scheme, none of Alice and Bob can retrieve any additional information in the noiseless scenario. Now, suppose that our claim is wrong i.e., Alice and Bob can pass all the tests mentioned in our scheme and later Alice can retrieve more data bits (than what she intends to know) in a single query or Bob can guess Alice's query indices with a more certain probability (than his intended probability). 



We now discuss this issue in the context of a particular form of non-i.i.d. attack, where a specific number of states are independently corrupted (more general attacks are also possible but these are outside the scope of this work). In this context, we will show that if some of the corrupted states are included during the testing phases, then there is some probability of being caught in the asymptotic limit.  
	
At the beginning of our scheme, the untrusted third party shares all the states with Alice and Bob. As in the \evp, both the parties choose the states randomly from the shared instances for the local tests at their end, the dishonest party can not guess beforehand the shared instances that the honest party will choose at his end for the local test. According to our assumption, the dishonest party can not manipulate the honest party's device once the protocol starts. So, to successfully pass the LocalCHSH test at the honest party's end, the shared states must be EPR pairs as specified in our scheme. This implies that the \evp~certifies all the states provided by the untrusted third party.
	
	
	We now explain these things more formally. Let us suppose that initially, the untrusted third party colludes with either the dishonest Alice or the dishonest Bob and shares either $\K_{\mathcal{A}}$ corrupted states in favour of Alice (let us denote this type of states as $\A$-type) or $\K_{\mathcal{B}}$ corrupted states in favour of Bob (let us denote this type of states as $\B$-type) among $\K$ shared states. So, while choosing randomly for the LocalCHSH test at honest Bob's end, the probability that a chosen state is of $\A$-type is $\frac{\K_{\mathcal{A}}}{\K}$. Similarly, for the LocalCHSH test at honest Alice's end, the probability that a chosen state is of $\B$-type is $\frac{\K_{\mathcal{B}}}{\K}$. Let us further assume that for the $\A$-type states, the value of the parameter $\mathcal{C}$ is $\mathcal{C}_{\A}$ (where $\mathcal{C}_{\A} = \mathcal{C} + \epsilon_{\A}$ such that $\epsilon_{\A} > 0$) and for the $\B$-type states, the value of the parameter $\mathcal{C}$ is $\mathcal{C}_{\B}$ (where $\mathcal{C}_{\B} = \mathcal{C} + \epsilon_{\B}$ such that $\epsilon_{\B} > 0$).
	
	Now, suppose that only Alice is dishonest and the third party supplies $\K_{\A}$ number of corrupted states (in favour of dishonest Alice) along with $(\K-\K_{\A})$ actual states. Then, in the localCHSH test at Bob's end, the probability that a chosen state is not of the $\A$-type is $\left(1-\frac{\K_{\A}}{\K}\right)$. One can easily check that this probability is also same for a chosen state in the final QPQ phase. As, dishonest Alice's aim is to gain as much additional data bits as possible in the final QPQ phase, she needs to choose the value of $\K_{\A}$ such that $(\K-\K_{\mathcal{A}})=c$ where $c$ is exponentially smaller than $\K$ (i.e., she will try to maximize the probability that a state chosen for the final QPQ phase is of the $\A$ type). Then, the probability that Bob will choose none of the corrupted states (i.e., the $\A$ type states) among his chosen $\frac{\gamma_1 \K}{2}$ states for the LocalCHSH test at his end is,
	
	\begin{align*}
	\left(1- \frac{\K_{\mathcal{A}}}{\K}\right)^{\frac{\gamma_1 \K}{2}} &= \left(\frac{c}{\K}\right)^{\frac{\gamma_1 \K}{2}}
	\end{align*}
	
	which is very small compared to $\K$. Similarly, whenever Bob is dishonest, the same thing can be shown for the LocalCHSH test at honest Alice's end. This implies that if the third party colludes with the dishonest party and supplies corrupted states then the probability that none of those corrupted states are chosen for the localCHSH test at the honest party's end is very small.
	
	In our scheme, we consider the ideal scenario where there are no channel noise. So for dishonest Alice, to successfully pass the LocalCHSH test at the honest Bob's end, the following relation must hold in the noiseless condition.
	
	\begin{align*}
	\frac{\K_{\A} \mathcal{C}_{\A}}{\K} +  \frac{(\K-\K_{\A}) \mathcal{C}}{\K} &= \mathcal{C} \\
	\K_{\A} \mathcal{C}_{\A} + (\K-\K_{\A}) \mathcal{C} &= \K \mathcal{C}\\
	\K_{\A} (\mathcal{C}_{\A} - \mathcal{C}) &= 0 
	\end{align*}
	
	Now, replacing the values of $\mathcal{C}_{\A}$ from the relation $\mathcal{C}_{\A} = \mathcal{C} + \epsilon_{\A}$, one can get,
	
	\begin{align}
	\label{relnsample}
	\K_{\A} \epsilon_{\A} &= 0 
	\end{align} 
	
	As the value of $\epsilon_{\A} > 0$, from this relation, one can easily conclude that in the noiseless scenario, the value of $\K_{\A}$ must be zero to successfully pass the LocalCHSH test at the honest Bob's end. Similarly, one can show that whenever Bob is dishonest, the value of $\K_{\B}$ must be zero to successfully pass the LocalCHSH test at the honest Alice's end. In practice, for finite number of samples, one can show that the values of $\K_{\A}$ and $\K_{\B}$ must be very small to successfully pass the local test at the honest party's end. 

Here, all the states are shared between the two parties before the start of the protocol and the dishonest party can not manipulate the honest party's device after the start of the protocol. As in this work, we focus on the {\it i.i.d.} scenario, it is straightforward to argue that either Alice and Bob abort the protocol with high probability in the asymptotic limit, or the LocalCHSH test certifies that the shared states involved in our QPQ scheme achieve the intended value of $\mathcal{C}$.
	
	

	The next DI testing is done in \bvp~where Bob and Alice perform distributed test to certify Bob's device. Here, one may think that if Bob is dishonest, then for the instances chosen in \bvp~and in \avp, he will measure in the actual measurement basis at his end to detect the fraudulent behaviour of Alice, and later for the instances to be used for the actual QPQ phase, he will measure in some different basis to guess the positions of Alice's known key bits. 
	
	From the results derived in lemma \ref{min_ent3} later, it is clear that if dishonest Bob wants to guess Alice's query indices with more certain probability then he must allow dishonest Alice to know more number of data bits in a single query. But this violates the assumption (more specifically the assumption 4) that none of the parties leak more information from their side to gain additional information from the other party. From the discussion in Lemma \ref{min_ent2}, it is also clear that for our scheme, the client Alice performs optimal strategy at her end. This implies that, for dishonest Alice, it is impossible to retrieve more data bits in a single query without manipulating the shared states or Bob's measurement device. Thus, to ensure that dishonest Alice is not getting any additional data bits, Bob must behave honestly in \bvp~to certify his device after the successful completion of \evp.
	
	In our scheme, before the \bvp, Bob generates a random bit for each of his qubits and measures his qubits accordingly. In the \bvp, Bob generates random bits for each of the Alice's qubits chosen for \bvp~and declare those bits so that Alice can measure her particles accordingly. As Bob behaves honestly in \bvp~(to restrict Alice from knowing additional data bits) and chooses all the inputs randomly for OBStest, there is no possibility that the inputs for OBStest are chosen according to some dishonest distribution.	From the analysis of Theorem \ref{cor5}, it is clear that if the inputs are chosen randomly then OBStest certifies that Bob's measurement device measures correctly in $\{\ket{0}, \ket{1}\}$ and $\{\ket{0'}, \ket{1'}\}$ basis for our proposed QPQ scheme.
	
	
	
	This implies that the successful completion of \evp~and \bvp~certifies that the shared states are EPR pairs and Bob's measurement device measures correctly for all the instances. This also implies that for all the remaining instances (that will be used for \avp~and in the actual QPQ phase), Alice has non-orthogonal qubits (i.e., either $\ket{0}$ or $\ket{0'}$ for $a_i = 0$ and either $\ket{1}$ or $\ket{1'}$ for $a_i = 1$) at her end.
	
	It is already mentioned that in our scheme, the client Alice performs optimal (POVM) measurement at her end to extract maximal number of data bits in a single query. So, after successful completion of \evp~and \bvp, Alice must behave honestly in \avp~to ensure that her measurement device is the optimal one. For this reason, Alice must measure her qubits accordingly as mentioned in KEYgen() and POVMtest() to certify her device. From the analysis of Theorem \ref{thm:povmtest0} and Theorem \ref{thm:povmtest1}, it is clear that the successful completion of \avp~certifies Alice's POVM device.
	
Note that in the proof of Theorem \ref{thm:povmtest0} and Theorem \ref{thm:povmtest1} in Appendix C (entitled {\it DI Testing of POVM Elements}), we have not imposed any dimension bound like the self-testing of POVM in a prepare and measure scenario in~\cite{TSVBB20}. So, the devices that perform a Neumark dilation of this mentioned POVM (i.e., the equivalent larger projective measurement on both the original state and some ancilla system instead of the actual POVM measurement) could still achieve the intended value of $\Omega$. But both of these operations produce the same output probabilities, which is sufficient for the purposes of this work.
	
	Hence from all these discussions, we can conclude the following-

\begin{corollary}
	\label{corrdisec}
	Either our DI proposal aborts with high probability in the asymptotic limit, or it certifies that the devices involved in our QPQ scheme achieve the intended values of $\mathcal{C}$, $\beta$, and $\Omega^0 (\text{or~} \Omega^1)$ in the LocalCHSH test, OBStest and POVMtest respectively.
\end{corollary}


	\subsubsection{\textbf{Database security against dishonest Alice}}
	
	In this subsection, we estimate the amount of raw key bits that dishonest Alice can guess in the \kep~of our proposed scheme.
	
	\begin{theorem}
		\label{min_ent1}
		After the \bvp, if Alice's measurement device is not tested then in the \kep, dishonest Alice can inconclusively (i.e., can't know the positions of the correctly guessed bits with certainty) retrieve (on average) at most $\left(\frac{1}{2}+ \frac{1}{2}\sin \theta\right)$ fraction of bits of the entire raw key.
	\end{theorem}
	
	\begin{proof} 
		After the key establishment phase, dishonest Alice ($\A^*$) and honest Bob ($\B$) share $kN$ raw key bits generated from the $kN$ copies of EPR pairs. The $i$-th copy of the state is given by $|\phi^+\rangle_{\A^*_i\B_i} = \frac{1}{\sqrt{2}}
		|00\rangle_{\A^*_i\B_i} + \frac{1}{\sqrt{2}}|11\rangle_{\A^*_i\B_i}$, where $i$-th subsystem of Alice and Bob is denoted by 
		$\A^*_i$ and $\B_i$ respectively. At Alice's side the reduced density matrix is of the form $$\rho_{\A^*_i} = 
		\tr_{\B_i}\left[|\phi^+\rangle_{\A^*_i\B_i}\langle{\phi^+}|\right] = \frac{\I_2}{2}.$$ At the beginning, Bob measures each of his part
		of the state $|\phi^+\rangle_{\A^*_i\B_i}$ in either $\{|0\rangle, |1\rangle\}$ basis or in $\{|0'\rangle,|1'\rangle\}$ basis. 
		The choice of the basis is completely random as this choice depends on the random raw key bit values chosen by Bob. Let $\rho_{\A^*_i|r_i}$ denotes the state at Alice's side after the choice of 
		Bob's measurement basis. For $r_i = 0$, we have, 
		
		\begin{align*}
		\rho_{\A^*_i|r_i =0} & = \tr_{\B_i}[\phi^+\rangle_{\A^*_i\B_i}\langle{\phi^+}|]\\
		& =  \tr_{\B_i}[\frac{1}{2}(|00\rangle + |11\rangle)_{\A^*_i\B_i}(\langle00| + \langle11|)]\\
		& = \frac{\I_2}{2}.
		\end{align*}
		
		Similarly, for $r_i =1$, we have, $\rho_{\A^*_i|r_i =1} = \frac{\I_2}{2} = \rho_{\A^*_i}$. This implies that $\rho_{\A^*_i|r_i} = \rho_{\A^*_i}$. In \bvp, Alice knows the declared $a_i$ values for all the instances. Let $\rho_{\A^*_i|a_i}$ denotes the state of Alice given the value 
		of $a_i$. According to the protocol,
		
		\begin{align*}
		\rho_{\A^*_i|a_i=0} & = \frac{1}{2} |0\rangle\langle0| + \frac{1}{2} |0'\rangle\langle0'|\\
		\rho_{\A^*_i|a_i=1} & = \frac{1}{2} |1\rangle\langle1| + \frac{1}{2} |1'\rangle\langle1'|.
		\end{align*}
		
		This implies that for a fixed $a_i = 0$ ($a_i = 1$) if Alice wants to guess the value of $r_i$ then she needs to distinguish 
		the state from the ensemble of states $\{(\frac{1}{2}|0\rangle\langle0|), (\frac{1}{2}|0'\rangle\langle0'|)\}$ 
		($\{(\frac{1}{2}|1\rangle\langle1|), (\frac{1}{2}|1'\rangle\langle1'|)\}$). In other words, whenever Bob measures his qubit and announces the bit $a_i = 0$, Alice knows that Bob gets either $\ket{0}$ or $\ket{0'}$. Similarly, when Bob 
		announces the bit $a_i = 1$, Alice knows that Bob gets either $\ket{1}$ or $\ket{1'}$. So, to retrieve the value of the original
		raw key bit, Alice needs to distinguish between the states $\ket{0}$ and $\ket{0'}$ or between the states $\ket{1}$ or $\ket{1'}$.
		
		Now if Alice's measurement device is not tested, then Alice can choose any measurement device at her side to distinguish the non-orthogonal states generated at her side. As it is known that non-orthogonal quantum states cannot be distinguished perfectly, Alice cannot guess the value of each raw key bit with certainty. This distinguishing probability has a nice relationship with the trace distance between the states in the ensemble \cite{Wilde17}. According to this relation we have,
		\begin{align*}
		\Pr_{\gs}[r_i|\rho_{\A^*_i|a_i=0}] & = \frac{1}{2}(1+\frac{1}{2}|||0\rangle\langle0| - |0'\rangle\langle 0'|||_1)\\
		& \leq \frac{1}{2}(1+\sqrt{1 - F(|0\rangle\langle0|, |0'\rangle\langle0'|)})\\
		& = \frac{1}{2}(1+\sin \theta) = \frac{1}{2} + \frac{1}{2}\sin \theta.
		\end{align*}
		
		One can check that $\Pr_{\gs}[r_i|\rho_{\A^*_i|a_i=0}] = \Pr_{\gs}[r_i|\rho_{\A^*_i|a_i=1}]$. This implies that if Alice is allowed to use any measurement device at her end after \bvp~then Alice can successfully retrieve the $i$-th raw key bit $r_i$ with probability at most $\left(\frac{1}{2} + \frac{1}{2}\sin \theta\right)$. As after \bvp, the qubits at Alice's side are all independent, dishonest Alice can inconclusively retrieve (on average) atmost $\left(\frac{1}{2} + \frac{1}{2}\sin \theta\right)$ fraction of bits of the entire raw key.
	\end{proof}

	{\it Note : The term `inconclusive' implies here that the client Alice can't predict the positions of the correctly guessed key bits with certainty. For example, whenever Alice tries to guess each of the key bits randomly, she can guess correctly for around half of the instances. However, she can't tell with certainty what are those instances for which she guesses correctly. 
	
	Now let us consider the operator $E=\{E_0,E_1,E_2\}$ where,
	
	\begin{eqnarray*}
		E_0 &\equiv& \frac{1}{\sin{\theta}} (\sin{\theta}\ket{0}-\cos{\theta}\ket{1})(\sin{\theta}\bra{0}-\cos{\theta}\bra{1})\\
		E_1 &\equiv& \frac{1}{\sin{\theta}}\ket{1}\bra{1}\\
		E_2 &\equiv& I - E_0 - E_1
	\end{eqnarray*}
	
	One can easily check that this operator $E=\{E_0,E_1,E_2\}$ is not a valid POVM as $E_2$ is not positive semi-definite. Let us consider the operator $E' = \{E'_0,E'_1\}$ where
	
	\begin{eqnarray*}
		E'_0 &\equiv& E_0 + \frac{E_2}{2}\\
		E'_1 &\equiv& E_1 + \frac{E_2}{2}
	\end{eqnarray*}
	
	Now, this is a valid POVM to distinguish $\ket{0}$ and $\ket{0'}= (\cos{\theta}\ket{0}+\sin{\theta}\ket{1})$. If a party consider the strategy that for the outcome $E'_0$, he consider the corresponding input qubit as $\ket{0}$ and $\ket{0'}$ otherwise, then one can check that this is the POVM corresponding to the optimal success probability (i.e., $\frac{1}{2}+\frac{\sin{\theta}}{2}$) in distinguishing $\ket{0}$ and $\ket{0'}$. However, the guessing outcome of this POVM is uncertain as the inconclusive element (the outcome which can't determine the state with certainty) $E_2$ is involved in both the elements $E'_0$ and $E'_1$ of the POVM $E'$. So, in the proof of Theorem \ref{min_ent1}, we refer the optimal guessing probability as inconclusive (i.e., uncertainty about the positions of the known key bits).}

	In theorem \ref{min_ent1}, we show that if Alice is allowed to choose any measurement device at her side then, on average, dishonest Alice can correctly retrieve at most around $\left(\frac{1}{2} + \frac{\sin{\theta}}{2}\right)$ fraction of bits of the entire raw key but she remains uncertain about the positions of those known bits. 
	

	However, in this DI proposal, dishonest Alice's ($\A^*$) main intension is to conclusively (i.e., with certainty about the positions of the correctly guessed key bits) retrieve as many raw key (as well as final key) bits as possible because otherwise she can't know which data bits she has retrieved correctly. For this reason, dishonest Alice has to perform the mentioned POVM measurement at her end to retrieve maximum number of raw key bits conclusively. Because of this, we can get a bound on the number of raw key bits that dishonest Alice can retrieve (on average) in this proposed DI-QPQ scheme.

		
		\begin{lemma}
			\label{min_ent2}
			After the \kep~of our proposed scheme, either the protocol aborts with high probability in the asymptotic limit, or dishonest Alice's strategy ($\A^*$) can retrieve (on average) $(1 - \cos{\theta})$ fraction of bits of the entire raw key.
		\end{lemma}

		\begin{proof}
			
			According to our proposal, after the \avp, the client Alice has $kN$ independent non-orthogonal qubits at her end. For each of these instances, dishonest Alice now tries to distinguish between the non-orthogonal states either $\ket{0}$ and $\ket{0'}$ (for $a_i = 0$) or $\ket{1}$ and $\ket{1'}$ (for $a_i = 1$). 
			
			
            In this regard, she chooses the measurement device $\{M^0_0, M^0_1,M^0_2\}$ when Bob announces $a_i=0$ and measurement device $\{M^1_0, M^1_1,M^1_2\}$ when Bob announces $a_i=1$.

            Whenever the outcome is $M^0_0$ ($M^1_0$), Alice concludes that the state is $\ket{0}$ ($\ket{1}$). If it is $M^0_1$ ($M^1_1$), she concludes that the state is $\ket{0'}$ ($\ket{1'}$). The guessing remains inconclusive (i.e., can't guess the outcome with certainty) only when the measurement outcome is $M^0_2$ ($M^1_2$). 

            In~\cite{iva87}, it is already mentioned that the maximum success probability in distinguishing two non-orthogonal states is ($1-\cos{\theta}$). From Theorem~\ref{correct}, we get that in our protocol, the success probability of Alice in guessing a key bit correctly and conclusively is also $(1-\cos{\theta})$. As Alice has to measure each of her qubits independently depending on the declared $a_i$ values, on average she can conclusively retrieve $(1 - \cos{\theta})$ fraction of bits of the entire raw key. This concludes the proof.
			
		\end{proof}

	For this proposed DI-QPQ scheme, the database contains $N$ number of data bits. Now relating definition \ref{data_priv} and equation \ref{finaldata}, we can derive the following bound on the value of $\tau$.
	
		
		\begin{corollary}
			\label{upper_lambda}
			In our full DI-QPQ proposal, for dishonest Alice and honest Bob, either the protocol aborts with high probability in the asymptotic limit, or dishonest Alice can guess on average $\tau$ fraction of bits of the final key, where
			\begin{equation}
			\noindent
			\tau \leq \left(1 - \cos{\theta}\right)^k 
			\end{equation}
			
			Replacing the value of $\left(1 - \cos{\theta}\right)^k$ with the upper bound mentioned in equation \ref{relation4}, we can get the following upper bound on the value of $\tau$.
			
			\begin{equation}
			\noindent
			\label{lambda}
			\boxed{\tau < \frac{2}{N} }
			\end{equation}
		\end{corollary}
		This relation shows that for this DI-QPQ proposal, $\tau$ is small compared to $N$.
		
		Now, we validate the probabilistic definition of data privacy for this proposed scheme and show that the probability $\Pr\left[|X-\mathbb{E}[X]| > \delta \wedge \text{protocol doesn't abort} \right]$ is negligible. More specifically, we will calculate the probability with which dishonest Alice can guess more than the expected number of final key bits (with a deviation more than the $\epsilon$ fraction of the expected number of final key bits).
		
		
		 The negligibility of the probability $\Pr\left[|X-\mathbb{E}[X]| > \delta \wedge \text{protocol doesn't abort} \right]$ can be shown using the properties of basic probability theory. Note that the probability $\Pr\left[|X-\mathbb{E}[X]| > \delta \wedge \text{protocol doesn't abort} \right]$ is upper bounded by both $\Pr\left[|X-\mathbb{E}[X]| > \delta\right]$ and $\Pr\left[\text{protocol doesn't abort} \right]$, according to the properties $\Pr[A \wedge B] \leq \Pr[A]$ and $\Pr[A \wedge B] \leq \Pr[B]$. As in our scheme, we consider the i.i.d. assumption, there will be two different subcases- 1) all the devices attain ideal values in all the testing phases (i.e., in LocalCHSHtest, OBStest and POVMtest) 2) all the devices don't attain ideal values in all the testing phases. 
		
		For the first subcase, from the correctness result (i.e., the value of $P_c$ for our scheme in equation \ref{boundpc}) and the DI security statement in Corollary \ref{corrdisec}, one can easily conclude that $\Pr\left[|X-\mathbb{E}[X]| > \delta \right] \leq negl(N)$ where $negl(N)$ denotes negligible in $N$. For the second subcase, by an analysis similar to the proof of Theorem~\ref{cor5} and from the DI security statement in Corollary \ref{corrdisec}, it can be concluded that $\Pr\left[\text{protocol doesn't abort} \right] \leq negl(N)$. This implies that for both of these two subcases, $\Pr\left[|X-\mathbb{E}[X]| > \delta \wedge \text{protocol doesn't abort} \right] \leq negl(N)$ (under the i.i.d. assumption).
		
		Although it is easy to derive the negligibility of the expression $\Pr\left[|X-\mathbb{E}[X]| > \delta \wedge \text{protocol doesn't abort} \right]$ for both the two subcases, in general for the second subcase, it is hard to derive the exact bound on the probability with which dishonest Alice can guess more than the expected number of final key bits. For our proposed scheme, as Alice performs optimal POVM measurement at her end, it is relatively easier to derive an upper bound on the parameter $P_d$ for our scheme because it is unlikely that dishonest Alice can retrieve more number of raw key bits (on average) by performaing any other measurements at her end.
		
		
		To derive an exact bound on the parameter $P_d$, like the discussion in Subsection B (entitled "parameter estimation for private query phase"), here also we consider that the random variable $X$ denotes the number of final key bits known by dishonest Alice and $E[X]$ is the expected value in honest scenario.
		
		Now, from the Chernoff-Hoeffding inequality~\cite{WH} mentioned in proposition \ref{propchernoff}, we can write the following,
		
		\begin{align}
		\noindent
		\label{pdval}
		\nonumber
		& \Pr\left[|X-\mathbb{E}[X]| \geq \delta_{CH} \wedge \text{protocol doesn't abort} \right] \\
		&\leq \exp(-2\delta_{CH}^2 N)
		\end{align}
		
		Here, we want to estimate the probability that the value of $X$ lie outside the error margin $\delta_{CH} = \epsilon \left(1 - \cos{\theta}\right)^k N$ from the expected value.
		
		From the relation in equation \ref{relation4}, it can be easily derived that whenever Bob chooses the value of $\theta$ such that $\left(1 - \cos{\theta}\right)^k =\frac{1}{N}$, the equation \ref{pdval} becomes,
		
\begin{align}
\noindent
\nonumber
& \Pr\left[|X-\mathbb{E}[X]| \geq \epsilon \wedge \text{protocol doesn't abort} \right] \\
&\leq \exp(-2 \epsilon^2 N)
\end{align}
		
		So, according to the definition \ref{data_priv}, in our proposed scheme, for dishonest Alice and honest Bob, the value of the parameter $P_d$ is upper bounded by,
		
		\begin{equation}
		\noindent
		\label{pdbound}
		\boxed{P_d \leq \exp(-2 \epsilon^2 N) }
		\end{equation}
		
		In practice, this probability is low as the value of $N$ is very large. So, for our proposed scheme, the probability that dishonest Alice can know more than the expected number of final key bits (with a deviation more than the $\epsilon$ fraction of the expected number of final key bits) and the protocol does not abort is very low.
		

	As an example, here we consider our scheme as a $1$ out of $2$ probabilistic oblivious transfer (i.e., $N=2$ and $k=1$) to evaluate the performance. From the expression \ref{pdbound}, if we consider $\epsilon = \frac{1}{2}$, we can conclude that for $1$ out of $2$ probabilistic oblivious transfer, the cheating probability of dishonest Alice in guessing more than the expected number of final key bits (which is $1$ in this case) will be upper bounded by
	
	\begin{equation}
	\label{12pdbound}
	P_d \leq e^{-1} \approx 0.368
	\end{equation}

	The comparative study between maximum inconclusive (i.e., the positions of the correct bits can't be guessed with certainty) success probability and maximum conclusive (i.e., the positions of the correct bits can be guessed with certainty) success probability is shown in figure \ref{comp}. From the figure, it is clear that the maximum inconclusive success probability outperforms maximum conclusive success probability for small values of $\theta$. 
	\begin{figure}[htbp]
		\includegraphics[scale=0.32]{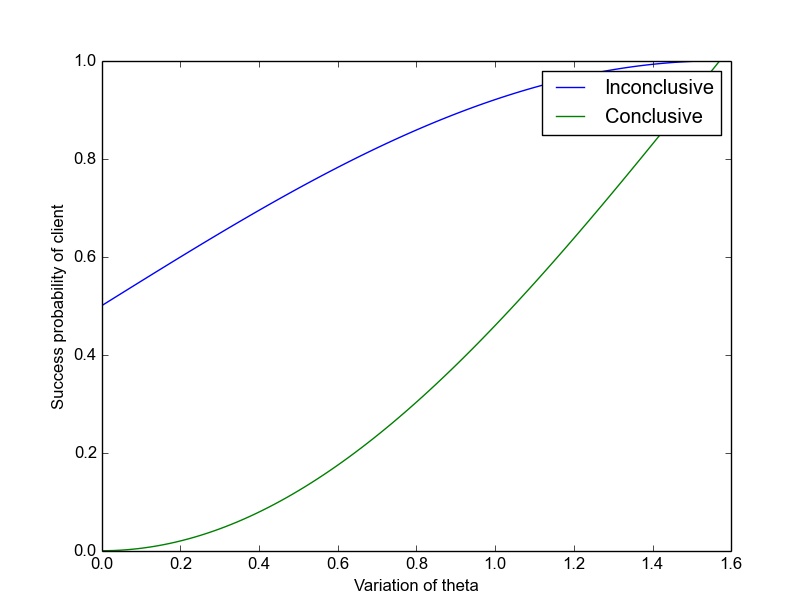}
		\caption{Comparison between maximum inconclusive and conclusive success probability of the client}
		\label{comp}
	\end{figure}
	
	\subsubsection{\textbf{User security against dishonest Bob}}


\begin{lemma}
	\label{min_ent3}
	In our proposed DI-QPQ scheme, after $l$ many queries to the $N$-bit database by Alice, dishonest Bob ($\B^*$) can successfully guess whether a particular index $i$ belongs to Alice's query index set $\ind_l$ (i.e., $i \in \ind_l$) with probability atmost $\frac{l}{N}$, i.e,
	\begin{equation}
	\noindent
	\Pr(\text{Bob guesses~} i\in \ind_l ) \leq \frac{l}{N} \nonumber
	\end{equation} 
\end{lemma}  

\begin{proof}
	At the \kep~of our proposed scheme, Alice does not broadcast anything about her measurement outcome. So, dishonest Bob has no information about Alice's measurement outcomes and her known key bits. Now, Alice queries $l$ many times to the database and retrieves $l$ many data bits. After these $l$ many queries, dishonest Bob will try to guess those query indices of Alice. As, Bob has no information about the known final key bits of Alice, he has to guess these $l$ many indices (out of the $N$ data bits) randomly.
	
	So, for any $i$-th data bit, dishonest Bob can guess whether $i \in \ind_l$ with probability atmost $\frac{l}{N}$. This completes the proof.
\end{proof}

From the proof of Lemma \ref{min_ent2}, it is clear that if $i$ denotes an index of the database then $$\Pr(\text{Bob guesses~}i\in \ind_l) \leq \frac{l}{N} $$ 

This implies that whenever Bob guesses a particular index from the data bits, the chosen index will be in Alice's query index set with probability around $\frac{l}{N}$. Now, here we assume that after $l$ many queries, Alice's query index set $\ind_l$ has $l$ many data bits and Alice chooses these $l$ bits independently. That means if Bob guesses $l$ many indices after the query phases then the expected value of the number of indices ($\ind_{\B^*}$) that dishonest Bob guesses correctly from the index set $\ind_l$ will be,


\begin{align}
\label{totalprob2}
E[\ind_{\B}] & = \Pr(\text{Bob guesses~}i\in \ind_l). l \nonumber\\
& \leq \frac{l^2}{N}
\end{align}

Now, comparing equation \ref{totalprob2} with the expression in the definition \ref{user_priv}, we can derive the following upper bound on the value of $\delta$ for our proposed scheme.

	
	\begin{corollary}
		\label{upper_delta}
		
		In our proposed QPQ scheme, for dishonest Bob and honest Alice, either the protocol aborts with high probability in the asymptotic limit, or dishonest Bob can guess on average $\delta$ fraction of indices from Alice's query index set $\ind_l$ where,
		\begin{equation}
		\noindent
		\label{delta1}
		\delta \leq \left(\frac{l}{N}\right)
		\end{equation}
	\end{corollary}
	Usually, in practice, the size of the database (i.e., $N$) is exponentially larger as compared to the size of Alice's query index set (i.e., $l$). More precisely, $N=l^n$, for some positive integer $n$. 
	
	Now, replacing this value of $N$ in equation \ref{delta1}, we can get the following bound on the value of $\delta$.

	\begin{equation}
	\noindent
	\label{delta}
	\boxed{\delta \leq \frac{1}{l^{(n-1)}}} 
	\end{equation} 
	
	where $n$ is a positive integer such that $n > 1$. This relation shows that for this full DI-QPQ proposal, $\delta$ is small compared to $l$.

Now, we validate the probabilistic definition of user privacy against dishonest Bob for our full DI proposal and derive the exact bound on the security parameter $P_u$. From the result of Lemma \ref{min_ent2}, it is clear that the probability with which dishonest Bob can guess whether a particular index $i$ belongs to Alice's query index set $\ind_l$ is upper bounded by $\frac{l}{N}$.

Here in Lemma~\ref{min_ent2}, the upper bound is calculated considering the scenario that the protocol does not abort. This implies that, 

\begin{align}
\noindent
\nonumber
& \Pr\left[\text{Bob guesses~}i\in \ind_l  \wedge \text{protocol doesn't abort} \right] \\
&\leq \frac{l}{N}
\end{align}

So, according to the definition \ref{user_priv}, for dishonest Bob and honest Alice, the value of the parameter $P_u$ in our scheme is upper bounded by,

\begin{equation}
\noindent
\label{puval}
\boxed{P_u \leq \frac{l}{N}} 
\end{equation}

	One can easily check that this probability is low as, in practice, the size of the database (i.e., $N$) is exponentially larger as compared to the size of Alice's query index set (i.e., $l$). So, for our proposed scheme, dishonest Bob can correctly guess whether a particular database index belongs to Alice's query index set with small probability.

	
	For performance evaluation, here also we consider our scheme as $1$ out of $2$ probabilistic oblivious transfer (i.e., $N=2$, $k=1$ and $l=1$). From expression \ref{puval}, we get that the value of $P_u$ for our scheme is upper bounded by,
	
	\begin{equation}
	\label{12puval}
	P_u \leq \frac{1}{2} \approx 0.5 
	\end{equation} \\

\section{Discussion and Conclusion}

Initially, all the QPQ schemes were proposed considering that the devices involved are trusted. Thus, a significant portion of the security issues depends on the functionality of the underlying devices. To remove such assumptions, Maitra et al~\cite{MPR17} first initiated the idea of DI in the QPQ domain by proposing a semi DI version of the QPQ scheme~\cite{Yang}. In this present draft, we move one step further and propose a novel QPQ scheme considering maximally entangled states with full DI certification to improve the robustness. We also discuss the optimality of the number of raw key bits that client Alice can retrieve for QKD-based QPQ schemes and show that in our proposed scheme, Alice retrieves the optimal number of raw key bits. Contrary to all the existing QPQ schemes which analyze the security issues considering certain eavesdropping strategies, here in this draft, we analyze the security of our scheme in the most general way considering all the attacks that preserve the correctness condition. We further manage to get upper bounds on the cheating probabilities for both the dishonest client and the dishonest server. As the recent QPQ schemes incorporate the idea of QKD, along with the other applications of oblivious transfer, QPQ may soon become a crucial near-term application of quantum internet.

	\section*{Appendix}
	
	Here we mention the procedure of choosing the initial sample size such that the two parties can certify the devices with desired accuracy and confidence. We also mention here the proofs of Theorem \ref{cor5}, Theorem \ref{thm:povmtest0} and Theorem \ref{thm:povmtest1} which confirms the functionality of the measurement devices involved in our protocol. We further show how ``up to unitary" devices preserve the correctness condition of our proposed scheme. In the first subsection, we show how one can choose the initial sample size for the protocol in practice. In the next subsection, we restate Theorem \ref{cor5} and mention the detail proof of the Theorem. We further restate the functionality of the POVM devices in the form of a two party game (namely POVMgame) and mention the detail proofs of Theorem \ref{thm:povmtest0} and Theorem \ref{thm:povmtest1} in the next subsection. Finally, in the last subsection, we show the correctness of our scheme whenever the devices are ``up to unitary" as compared to the original devices. 
	
		
		\section*{Appendix A : Choice of Initial Sample Size in Practice}

		In this section, we discuss how Bob and Alice choose the initial sample size required for the proposed DI-QPQ scheme. In practice, Alice and Bob have to allow some deviation (from the actual value of the parameter because of finite number of samples) in each testing phase to certify the devices.

		It is well-known that the approximate number of samples required to distinguish two events having probabilities $p$ and $p(1+\epsilon)$ (for small $\epsilon$) is $O(\frac{1}{p\epsilon^2})$. One may require approximately $\frac{64}{p\epsilon^2}$ samples to achieve a confidence of more than 99\% in distinguishing these two events. A more involved expression of the sample size is recently derived in~\cite{BM18} using Chernoff-Hoeffding~\cite{WH} bound which is stated in proposition \ref{propchernoff}. 
		
		
		For the testing phases mentioned in our proposed scheme, we consider $X_i = 1$ whenever Bob and Alice win the $i$-th instance and $X_i = 0$ otherwise. Now if we consider $\mathbb{E}[X_i] = p$ and want to estimate the success probability $p$ within an error margin of $\epsilon p$ and confidence $1 - \eta$, then from the result mentioned in~\cite{BM18}, we can write that the required sample size $m_\text{req}$ will be,
		
		\begin{equation}
		\label{reqsamp}
		m_\text{req} \geq \frac{1}{2\epsilon^2 p^2} \ln \frac{1}{\eta}
		\end{equation}
		
		From this expression of $m_\text{req}$, Bob and Alice can estimate the expected number of samples required for a particular testing phase to certify a device with certain accuracy and confidence. 
		
		Now to ensure that Bob and Alice get the expected number of samples in each phase (to conclude with certain accuracy and confidence), they choose the total initial sample size (i.e., the value of $\K$) as follows-
		
		\begin{itemize}
			\item Before the start of the protocol, Alice and Bob (based on the protocol description) calculate the minimum number of samples required (according to the expression in inequality \ref{reqsamp}) in each testing phase to conclude with chosen accuracy and confidence.
			\item Then they choose the value of $k$ to calculate the total number of samples required in {\it private query phase}.
			\item At last, they sum up all these number of samples required in each testing phase along with the number of samples required in private query phase to calculate the total initial sample size.
			\item After getting the initial sample size, Bob and Alice proceed to each of the testing phases (according to the description of the protocol), select the required number of samples randomly from the shared instances and check whether the value of a predefined parameter lies within the interval $[V-\epsilon p, V+\epsilon p]$ where $V$ is the actual value of the parameter obtained for asymptotically large number of samples. If this is the case, then with accuracy $\epsilon p$ and chosen confidence $(1-\eta)$, they conclude that the devices behave accordingly.
		\end{itemize}
		
		As an example, here we demonstrate the method of choosing samples for the first phase namely source device verification phase. Before the start of the protocol, Bob and Alice choose the accuracy and confidence parameter for this phase with which they want to certify the source device and let $n_1$ be the required number of samples. Now, similar to this \evp, they calculate the required number of samples for the other phases also and from that calculate the required number of total initial samples $\K$.
		
		Bob and Alice then calculate the value of $\gamma_1$ such that,
		\begin{align*}
		n_1 &= \gamma_1 \K
		\end{align*} 
		
		After getting the value of $\gamma_1$, Bob first chooses $\frac{\gamma_1 \K}{2}$ number of samples randomly from the $\K$ shared states and then from the rest $\left(\K-\frac{\gamma_1 \K}{2}\right)$ number of samples, Alice randomly chooses $\frac{\gamma_1 \K}{2}$ number of samples. They then discuss their chosen instances publicly, get the qubits from the other party and perform LocalCHSHtest for their chosen $\gamma_1 \K$ samples. In this similar way, they choose the samples for the remaining testing phases.
		
		Note that this is a particular way of choosing samples that we demonstrate here from the several other possibilities. It is needless to say that one may follow any other strategies for choosing samples in different testing phases.
		

	\section*{Appendix B : Statement and Proof of Theorem~\ref{cor5}}

	\emph{Theorem \ref{cor5}:~	In OBStest, either Bob's measurement devices achieve $\beta = \frac{1}{\sqrt{2}|(\cos{\theta}-\sin{\theta})|}$ (i.e., his devices measure correctly in $\{\ket{0}, \ket{1}\}$ and $\{\ket{0'}, \ket{1'}\}$ basis where $\ket{0'} = (\cos \theta\ket{0} + \sin \theta\ket{1})$, $\ket{1'} = (\sin \theta\ket{0} - \cos \theta\ket{1})$), or the protocol aborts with high probability in the asymptotic limit.\\}

	\emph{Proof:}~Suppose, Alice's measurement operators are $\{A^s_c\}_{s,c \in \{0,1\}}$, corresponding to the input $s$ and output $c$. Similarly, Bob's measurement operators are $\{B^r_b\}_{r,b \in \{0,1\}}$, corresponding to the input $r$ and output $b$. This implies that Alice's observable, corresponding to the input $s \in \{0,1\}$ is,
	
	\begin{equation}
	\label{eq:aobs}
	A_s = \sum_{c \in \{0,1\}}(-1)^cA^s_c. 
	\end{equation}
	Similarly, Bob's observable corresponding to the input $r \in \{0,1\}$ is,
	\begin{equation}
	\label{eq:bobs}
	B_r = \sum_{b \in \{0,1\}}(-1)^bB^r_b. 
	\end{equation}
	
	Note that, in the OBStest, the fraction $\beta$ is being computed as follows,
	
	\begin{align}
	\label{eq:beta}
	\beta &= \frac{1}{4} \sum_{s,r,c,b \in \{0,1\}}(-1)^{d_{srcb}}\alpha^{1\oplus s} \langle \phi_{\A\B}| A^s_c \otimes B^r_b |\phi_{\A\B}\rangle\\
	& = \frac{1}{4} \langle \phi_{\A\B}| \wa |\phi_{\A\B}\rangle, 
	\end{align}
	
	where $\wa := \left(\sum_{s,r,c,b \in \{0,1\}}(-1)^{d_{srcb}}\alpha^{1\oplus s}  A^s_c \otimes B^r_b  \right)$ which is the operator corresponding to the OBStest. We can also rewrite the expression of $\wa$ in the following way,
	
	\begin{align}
	\nonumber
	\wa  & =  \left(\sum_{r,c,b \in \{0,1\}}(-1)^{d_{srcb}}\alpha  A^0_c \otimes B^r_b \right) + \\ \nonumber
	&  \left(\sum_{r,c,b \in \{0,1\}}(-1)^{d_{srcb}} A^1_c \otimes B^r_b\right) \\ \label{eq:w}
	& = \wa^0 + \wa^1,
	\end{align}
	
	where $\wa^0 :=  \left(\sum_{r,c,b \in \{0,1\}}(-1)^{d_{srcb}}\alpha  A^0_c \otimes B^r_b \right) $ and $\wa^1 :=\left(\sum_{r,c,b \in \{0,1\}}(-1)^{d_{srcb}} A^1_c \otimes B^r_b \right)$. Note that, we can simplify further the expression of $\wa^0$ in following way,
	
	\begin{align*}
	\wa^0 &= \sum_{r,c,b \in \{0,1\}}(-1)^{d_{srcb}}\alpha  A^0_c \otimes B^r_b \\
	& =  \sum_{\substack{r,c,b \in \{0,1\}\\c\oplus b = 0}}\alpha  A^0_c \otimes B^r_b  - \sum_{\substack{r,c,b \in \{0,1\}\\c\oplus b \neq 0}}\alpha  A^0_c \otimes B^r_b \\
	& = \alpha (A^0_0 \otimes B^0_0 + A^0_0 \otimes B^1_0 + A^0_1 \otimes B^0_1 + A^0_1 \otimes B^1_1) - \\
	& \alpha (A^0_0 \otimes B^0_1 + A^0_0 \otimes B^1_1 + A^0_1 \otimes B^0_0 + A^0_1 \otimes B^1_0)\\
	& = \alpha [A^0_0 \otimes (B^0_0 - B^0_1) - A^0_1 \otimes (B^0_0 - B^0_1) + \\
	&A^0_0 \otimes (B^1_0 - B^1_1) - A^0_1 \otimes (B^1_0 - B^1_1)]\\
	& = \alpha[(A^0_0 - A^0_1)\otimes (B^0_0 - B^0_1) + (A^0_0 - A^0_1)\otimes (B^1_0 - B^1_1) ]\\
	& = \alpha(A^0_0 - A^0_1)\otimes [(B^0_0 - B^0_1) +(B^1_0 - B^1_1)].
	\end{align*}
	By substituting the values of $(A^0_0 -A^0_1)$, $(B^0_0 - B^0_1)$ and $(B^1_0 - B^1_1)$ from equation \ref{eq:aobs} and equation \ref{eq:bobs} on the right-hand side of the above expression we get,
	
	\begin{equation}
	\label{eq:w0}
	\wa^0 = \alpha A_0 \otimes (B_0 + B_1).
	\end{equation}
	
	Using similar approach we get the following simplified version of the expression $\wa^1$.
	
	\begin{equation}
	\label{eq:w1}
	\wa^1 = A_1 \otimes (B_0 - B_1).
	\end{equation}
	
	By substituting the values of $\wa^0$ and $\wa^1$ from equation \ref{eq:w0} and equation \ref{eq:w1} to equation \ref{eq:w} we get,
	
	\begin{equation}
	\label{eq:w_simp}
	\wa = \alpha A_0\otimes (B_0 +B_1) + A_1 \otimes (B_0 - B_1).
	\end{equation}  
	Note that, the right-hand side of this OBStest operator $\wa$ is exactly same as the tilted CHSH operator, described in \cite{kan17}. 
	
	So, the expression of $\wa^2$ can be written as
	
	\begin{align*}
	\wa^2 & = \alpha^2 A_0^2 \otimes (B_0^2 + B_1^2 + \{B_0,B_1\}) \\
	&~~~~+ A_1^2 \otimes (B_0^2 + B_1^2 - \{B_0,B_1\})\\
	&= (\alpha^2A_0^2 + A_1^2 + \alpha \{A_0,A_1\}) \otimes B_0^2\\ 
	&~~~~+ (\alpha^2A_0^2 + A_1^2 - \alpha \{A_0,A_1\}) \otimes B_1^2\\
	&~~+ (\alpha^2 A_0^2-A_1^2) \otimes \{B_0,B_1\}- \alpha [A_0,A_1] \otimes [B_0,B_1].
	\end{align*}
	
	Using the property $A_j^2 \leq \id$, we can rewrite this expression as,
	
	\begin{align*}
	\wa^2 &\leq [(\alpha^2 +1).\id + \alpha \{A_0,A_1\}] \otimes B_0^2\\ 
	&~~~~+ [(\alpha^2 +1).\id - \alpha \{A_0,A_1\}] \otimes B_1^2\\
	&~~+ \id \otimes (\alpha^2 - 1) \{B_0,B_1\} - \alpha [A_0,A_1] \otimes [B_0,B_1].
	\end{align*}
	
	Since $-2.\id \leq \{A_0,A_1\} \leq 2.\id$, we have,
	
	\begin{align*}
	[(\alpha^2 +1).\id \pm \alpha \{A_0,A_1\}] &\geq 0
	\end{align*}
	
	We can use the property $B_k^2 \leq \id$ and get the following simplified expression
	
	\begin{align*}
	\wa^2 &\leq 2(\alpha^2+1).\id \otimes \id + \id \otimes (\alpha^2 - 1) \{B_0,B_1\} \\
	&- \alpha [A_0,A_1] \otimes [B_0,B_1]
	\end{align*}
	
	We can further upper bound the commutators by their matrix modulus and use the relation $|[A_0,A_1]|\leq 2.\id$ to get the following expression
	
	\begin{align}
	\label{expwa2}
	\wa^2 &\leq 2(\alpha^2+1).\id \otimes \id + T_{\alpha} \otimes \id
	\end{align}	
	
	where $T_{\alpha}:=(\alpha^2-1)\{B_0,B_1\}+2\alpha|[B_0,B_1]|$
	
	Now the expression of $T_{\alpha}$ can also be upper bounded by upper bounding the anti commutators by its matrix modulus. So, the value of $T_{\alpha}$ will be upper bounded by,
	
	\begin{align*}
	T_{\alpha} &\leq (\alpha^2-1)|\{B_0,B_1\}|+2\alpha|[B_0,B_1]|
	\end{align*}
	
	Again one can easily check that, 
	
	\begin{align}
	\label{b0b1}
	~ &~ |\{B_0,B_1\}|^2 + |[B_0,B_1]|^2 \nonumber\\
	~ &= |B_0B_1+B_1B_0|^2 + |B_0B_1-B_1B_0|^2 \nonumber \\ 
	~ &= (B_0B_1+B_1B_0)^{\dagger}(B_0B_1+B_1B_0) \nonumber\\
	&+(B_0B_1+B_1B_0)^{\dagger}(B_0B_1+B_1B_0)\nonumber \\
	~ &= 2 (B_0B_1)^{\dagger}(B_0B_1)+2(B_1B_0)^{\dagger}(B_1B_0)
	\end{align}
	
	Let us consider that the measurement operators are projective i.e., $(A^s_c)^2=A^s_c$ and $(B^r_b)^2=B^r_b$. Now for the projectors $B^0_0$ and $B^0_1$, $(B^0_0+B^0_1)=\id$. From this relation we can write,
	
	\begin{align*}
	(B^0_0+B^0_1)(B^0_0+B^0_1)^{\dagger} &= \id\\
	B^0_0.{B^0_0}^{\dagger}+B^0_0.{B^0_1}^{\dagger}+B^0_1.{B^0_0}^{\dagger}+B^0_1.{B^0_1}^{\dagger} &= \id\\
	(B^0_0+B^0_1)+(B^0_0.{B^0_1}^{\dagger}+B^0_1.{B^0_0}^{\dagger})&=\id
	\end{align*}
	
	This implies,
	
	\begin{align*}
	(B^0_0.{B^0_1}^{\dagger}+B^0_1.{B^0_0}^{\dagger})&= 0
	\end{align*}
	
	Now $B_0=(B^0_0-B^0_1)$. From this we can get,
	
	\begin{align*}
	B_0B_0^{\dagger}&=(B^0_0-B^0_1)(B^0_0-B^0_1)^{\dagger}\\
	~ &= B^0_0.{B^0_0}^{\dagger}-B^0_0.{B^0_1}^{\dagger}-B^0_1.{B^0_0}^{\dagger}+B^0_1.{B^0_1}^{\dagger}\\
	~ &= (B^0_0+B^0_1)-(B^0_0.{B^0_1}^{\dagger}+B^0_1.{B^0_0}^{\dagger})\\
	~ &= \id + 0 = \id
	\end{align*}
	
	Similarly, it can be shown that, $B_1B_1^{\dagger}=B_1^{\dagger}B_1=\id$.
	
	So, from equation \ref{b0b1}, we can write that for unitary observables $B_0$ and $B_1$,
	
	\begin{align*}
	|\{B_0,B_1\}|^2 + |[B_0,B_1]|^2 &= 2 (B_0B_1)^{\dagger}(B_0B_1)\\
	&+2(B_1B_0)^{\dagger}(B_1B_0)\\
	~ &= 2\id + 2\id =4\id
	\end{align*}
	
	This implies, 
	
	\begin{align*}
	|\{B_0,B_1\}| &= \sqrt{4.\id - |[B_0,B_1]|^2 }
	\end{align*}
	
	So, the simplified expression of $T_{\alpha}$ will be of the form
	
	\begin{align*}
	T_{\alpha} &= (\alpha^2-1)\sqrt{4.\id - |[B_0,B_1]|^2 }+2\alpha|[B_0,B_1]|
	\end{align*}
	
	This is the maximum value of $T_{\alpha}$ and here $T_{\alpha}$ attains this maximum value because of projective observables. Now one can easily check that the value of $|[B_0,B_1]|$ which maximizes the value of $T_{\alpha}$ is $|[B_0,B_1]|=\frac{4\alpha}{(\alpha^2 + 1)}.\id$ and the corresponding value of $T_{\alpha}$ is $2(\alpha^2+1).\id$. This implies that, 
	
	\begin{align*}
	T_{\alpha} &= 2(\alpha^2+1).\id
	\end{align*}

	From this value of $T_{\alpha}$ and from the expression of $\wa^2$ mentioned in equation \ref{expwa2}, we can easily write that the value of $\wa$ is upper bounded by the following quantity-
	\begin{equation}
	\label{proofwa}
	\wa \leq\sqrt{2(\alpha^2+1)\id \otimes \id + T_{\alpha} \otimes \id}
	\end{equation}
	where $T_{\alpha} = 2(\alpha^2+1).\id$.
	
	Now, the value $\beta$ obtained in OBStest of our algorithm can be written alternatively as $\beta=\frac{\tr(\wa\rho_{\A\B})}{4}$ where $\rho_{\A\B}$ is the density matrix representation of the shared states $\ket{\phi}_{\A\B}$ i.e., $\rho_{\A\B} = \ket{\phi}_{\A\B}\bra{\phi}$. From this expression of $\beta$, one can easily derive that the value of $\beta^2$ is upper bounded by the following quantity,
	
	\begin{align*}
	\beta^2 &\leq \frac{\tr(\wa^2\rho_{\A\B})}{16}
	\end{align*}

	Now if we assume $t_{\alpha} : = \frac{1}{4} \tr(T_{\alpha}\rho_{\B}) - \frac{1}{2}(\alpha^2 -1)$ (where $\rho_{\B}$ is the reduced state at Bob's side) then using this value of $t_{\alpha}$ along with the value of $\wa$ obtained from expression \ref{proofwa} and the upper bound on the value of $\beta^2$, we can write that the $\beta$ value mentioned in OBStest is upper bounded by the following quantity,
	\begin{equation}
	\label{eq:upp_bet}
	\beta \leq \frac{\sqrt{\alpha^2 + t_{\alpha}}}{2},
	\end{equation}

	Now here, the observables are projective (i.e., $B_j^2 =\id$) and the anti commutator $\{B_0,B_1\}$ is a positive semi definite operator. Since we have already shown that the value of the anti-hermitian operator $|[B_0,B_1]|$ is $|[B_0,B_1]|=\frac{4\alpha}{(\alpha^2 + 1)}.\id$ for the maximum value of $T_{\alpha}$, the spectral decomposition of $[B_0,B_1]$ can be written as,
	
	\begin{align*}
	[B_0,B_1] &= \frac{4\alpha.i}{(\alpha^2 + 1)} (P_+ - P_-)
	\end{align*} 
	
	for some orthogonal projectors $P_+$ and $P_-$ such that $(P_+ + P_-)=\id$. As it is well-known that for projective observables, the commutator holds the property $B_0[B_0,B_1]B_0=-[B_0,B_1]$, we can easily conclude that $B_0P_{\pm}B_0=P_{\mp}$. Let us consider that $\{\ket{e^0_j}\}_j$ is an orthonormal basis for the support of $P_+$ and $\{\ket{e^1_j}\}_j$ is an orthonormal basis for the support of $P_-$ where $\ket{e^1_j}=B_0\ket{e^0_j}$. We define the unitary operator $U_0$ as
	
	\begin{align*}
	U_0\ket{e^d_j}= \frac{1}{\sqrt{2}}[\ket{0}+ (-1)^di\ket{1}] \ket{j}
	\end{align*}  
	
	for $d \in \{0,1\}$. Then we can easily verify that 
	
	\begin{align*}
	U_0[B_0,B_1]U_0^\dagger = \frac{4\alpha.i}{(\alpha^2 +1)} \sigma_Y \otimes \id
	\end{align*}
	
	Since $\{\id,\sigma_X,\sigma_Y,\sigma_Z\}$ constitute an operator basis for linear operators acting on $\mathbb{C}^2$, without loss of generality we can write
	
	\begin{align*}
	U_0B_0U_0^\dagger = \id \otimes K_0+\sigma_X \otimes K_x+\sigma_Y \otimes K_y+\sigma_Z \otimes K_z
	\end{align*}
	
	for some hermitian operator $K_0,K_x,K_y,K_z$. For projective observable $B_0$, one can easily check that $\{B_0,[B_0,B_1]\}=0$. This relation satisfies only when $K_0=K_y=0$. As $B_0^2=\id$, $K_x$ and $K_z$ must satisfy the relation
	
	\begin{align*}
	K_x^2 + K_z^2=\id~~~ \text{and}~~~ [K_x,K_z]=0
	\end{align*}

	So, we can easily write $K_x$ and $K_z$ in the following form.
	
	\begin{align*}
	K_x &= \sum_{j}\sin{2\gamma_j}\ket{j}\bra{j}\\
	K_z &= \sum_{j}\cos{2\gamma_j}\ket{j}\bra{j}
	\end{align*}
	
	for some angle $\gamma_j$ and some orthonormal basis $\{\ket{j}\}$. This implies that
	
	\begin{align*}
	U_0B_0U_0^\dagger &= \sigma_X \otimes K_x + \sigma_Z \otimes K_z\\
	&= \sum_{j} (\sin{2\gamma_j}\sigma_X + \cos{2\gamma_j}\sigma_Z) \otimes \ket{j}\bra{j}
	\end{align*}
	
	We now consider the following controlled unitary to align the qubit observables.
	
	\begin{align*}
	U_1 &= \sum_{j}\exp(i\gamma_j.\sigma_Y)\otimes\ket{j}\bra{j}
	\end{align*}
	
	Now for this defined unitary operator, one can easily check that
	
	\begin{align*}
	U_1U_0B_0U_0^\dagger U_1^\dagger &= \sigma_Z \otimes \id \\
	U_1U_0[B_0,B_1]U_0^\dagger U_1^\dagger &= \frac{4\alpha.i}{(\alpha^2 +1)} \sigma_Y \otimes \id
	\end{align*}
	
	Like observable $B_0$, an analogous reasoning can also be applied for observable $B_1$ and from that, without loss of generality we can write
	
	\begin{align*}
	U_1U_0B_1U_0^\dagger U_1^\dagger = \sigma_X \otimes K_x' +  \sigma_Z \otimes K_z'
	\end{align*}
	
	Since the commutators are positive semi definite and the observables are projective, we can easily check that
	
	\begin{align*}
	\{B_0,B_1\} = |\{B_0,B_1\}| &= \sqrt{4.\id - |[B_0,B_1]|^2}\\
	&= \frac{2(\alpha^2-1)}{(\alpha^2 + 1)}.\id
	\end{align*}
	
	Now we define $2\theta := \arcsin\left(\frac{\alpha^2-1}{\alpha^2 +1}\right) \in [0,\frac{\pi}{2}]$. From this relation, imposing consistency on the anti commutator, we get,
	
	\begin{align*}
	K_z' &= \sin{2\theta}.\id
	\end{align*}
	
	On the other hand, imposing consistency on the commutator, we get,
	
	\begin{align*}
	K_x' &= \cos{2\theta}.\id
	\end{align*}
	
	Now, from the relation $2\theta := \arcsin\left(\frac{\alpha^2-1}{\alpha^2 +1}\right)$, we can get the value of $\alpha$ which is 
	
	\begin{align*}
	\alpha &= \frac{(\cos{\theta}+\sin{\theta})}{|(\cos{\theta}-\sin{\theta})|}
	\end{align*}

	For this value of $\alpha$, we can easily derive that $t_{\alpha}=1$. This implies that the simplified expression for $\beta$ is,
	\begin{equation}
	\label{eq:bet_val}
	\beta = \frac{\sqrt{1+\alpha^2}}{2}
	\end{equation} 
	
	where $\alpha=\frac{(\cos{\theta}+\sin{\theta})}{|(\cos{\theta}-\sin{\theta})|}$. Now from this value of $\alpha$, we can derive the value of $\sqrt{1+\alpha^2}$ which is,
	
	\begin{equation}
	\sqrt{1+\alpha^2} = \frac{\sqrt{2}}{|(\cos{\theta}-\sin{\theta})|}
	\end{equation}
	
	So, the value of $\beta$ corresponding to these observables $B_0$ and $B_1$ will be,
	
	\begin{equation}
	\beta = \frac{1}{\sqrt{2}|(\cos{\theta}-\sin{\theta})|}
	\end{equation}
	
	If we consider $U_{\B}=U_0^\dagger U_1^\dagger$ then the observables $B_0$ and $B_1$ will be of the form
	
	\begin{align*}
	B_0 &= U_{\B} (\sigma_Z \otimes \id) U_{\B}^\dagger\\
	B_1 &= U_{\B} (\cos{2\theta}\sigma_X + \sin{2\theta}\sigma_Z \otimes \id) U_{\B}^\dagger
	\end{align*}
	
	This implies that in the OBStest, if $\beta$ is equal to $\frac{1}{\sqrt{2}|(\cos{\theta}-\sin{\theta})|}$, then the corresponding observables of Bob are same as the one described in the OBStest. This concludes the proof.

	\section*{Appendix C : DI Testing of POVM Elements}
	In the QPQ protocol, Alice needs to make sure her measurement device works properly, i.e, she should be able to distinguish between $|0\rangle$ ($|1\rangle$) and $|0'\rangle$ ($|1'\rangle$) with certainty for (on average) around $(1- \cos \theta)$ fraction of instances, where, $|0'\rangle = \cos \theta |0\rangle + \sin \theta |1\rangle$ ($|1'\rangle = \sin \theta |0\rangle - \cos \theta |1\rangle$). Let $M^0 = \{M^0_0,M^0_1,M^0_2\}$ ($M^1 = \{M^1_0,M^1_1,M^1_2\}$) the set of Alice's POVMs, which distinguishes the states $\{|0\rangle, |0'\rangle\}$ ($\{|1\rangle, |1'\rangle\}$). Here we show that if the input states are of the form $|0\rangle$ ($|1\rangle$) or $|0'\rangle$ ($|1'\rangle$) and Alice manages to distinguish the states with certainty for (on average) around $(1-\cos \theta)$ fraction of instances then $M^0_i = D^0_i$ ($M^1_i = D^1_i$) for $i \in \{0,1,2\}$. In order to prove this, here we first represent the interactions between Bob and Alice in the proposed DI-QPQ protocol in the form of a game, called POVMgame($M^y,y$) for better understanding, where the agent $A_1$ represents Bob and the agent $A_2$ represents Alice. The game is as follows,
	
	\RestyleAlgo{boxruled}
	\LinesNumbered
	\IncMargin{-1.5em}
	\begin{algorithm}[htbp]
		\begin{itemize}
			\item $A_1$ declares $y$ whenever the state at his side (and also at $A_2$'s side) is either $\rho^y_x$ or $\rho^y_{x \oplus 1}$ for the randomly chosen $x$ values (i.e., for $x \in_R \{0,1\}$), where $\rho^0_0 = |0\rangle\langle0|$, $\rho^0_1 = |0'\rangle \langle 0'|$, $\rho^1_0 = |1\rangle\langle1|$ and $\rho^1_1 = |1'\rangle\langle1'|$.
			\item $A_2$ measures her state (which is either $\rho^y_x$ or $\rho^y_{x \oplus 1}$) using the POVM $M^y$ (where $M^y=\{M^y_0, M^y_1, M^y_2\}$) and sends the outcome $b \in \{0,1,2\}$ to $A_1$.
			\item $A_2$ wins if and only if, $\Omega^y=\sum_{b,x \in\{0,1\}} (-1)^{b\oplus x} \tr[M^y_b\rho^y_x] = \frac{2 \sin^2 \theta}{1+\cos \theta}$.
		\end{itemize}
		\caption{POVMgame($M^y,y$)}
		\label{app:POVMtest} 
	\end{algorithm}
	\DecMargin{-1.5em}

	\begin{theorem}
		\label{app_thm:max_winD}
		In POVMgame($M^y,y$), if $A_1$ chooses $y=0$ and the states at $A_2$'s end are $\rho^0_0 = |0\rangle\langle0|$ and $\rho^0_1 = |0'\rangle \langle 0'|$ and if $A_2$ manages to win the game, i.e., $\Omega^0 = \frac{2 \sin^2 \theta}{1+\cos \theta}$, then this implies, $A_2$'s measurement devices are of the following form (up to a global unitary),
		
		\begin{align}
		M^0_0 & = \frac{1}{(1+\cos \theta)} (|1'\rangle \langle1'|)\\
		M^0_1 & = \frac{1}{(1+\cos \theta)} (|1\rangle \langle1|)\\
		M^0_2 & = \id - M^0_0- M^0_1,
		\end{align}
		
		where, $|1'\rangle = \sin \theta |0\rangle - \cos \theta |1\rangle$.
		
	\end{theorem}

	\begin{proof}
		In the POVMgame($M^y,y$), $A_2$ applies $M^0$ on a single qubit state $\rho^0_x$ (where $x \in_R \{0,1\}$). So, without any loss of generality we can assume that $M^0_i \in M^0$ has the following form,
		
		\begin{equation}
		M^0_i = \lambda^0_i(\id + \m^0_i.\asig),
		\end{equation}
		where $\m^0_i = [m^0_{i0},m^0_{i1},m^0_{i2}]$ and it is the Bloch vector with length at most one, $\asig = [\sigma_X, \sigma_Y, \sigma_Z]$ are the Pauli matrices and $\lambda_i \geq 0$. 
			In this case, one may wonder how we can fix the dimension of $M^0_i$ here in the proof in DI scenario? The answer to this question is that here we are able to fix the dimension of $M^0_i$ and choose this particular general form because of the tests mentioned earlier in the source device verification phase (corresponding result mentioned in Corollary \ref{cor3}) and DI testing phase for Bob's measurement device (corresponding result mentioned in Theorem \ref{cor5}) which certifies that the states shared between Alice and Bob are EPR pairs (up to a unitary) and after Bob's projective measurements, the reduced states at Alice's side are one qubit states. 
		Now, the condition $\sum_{i=0}^2 M^0_i = \id$ leads us to the following relations,
		
		\begin{align}
		\label{app_eq:lamb}
		\sum_{i=0}^2 \lambda^0_i &= 1\\ \label{app_eq:lambm}
		\sum_{i=0}^2 \lambda^0_i\m^0_i &= 0.
		\end{align}
		
		In terms of Bloch vector we can rewrite $\rho^0_0, \rho^0_1$ in following way,
		
		\begin{align}
		\rho^0_0 &= \frac{1}{2}(\id + \sigma_Z)\\
		\rho^0_1 &= \frac{1}{2}(\id + \sin 2\theta \sigma_X + \cos 2\theta \sigma_Z).
		\end{align}
		
		In the POVMgame($M^y,y$) if $A_2$ would like to maximizes her winning probability then she needs to maximize the following expression,
		
		\begin{equation}
		\Omega^0 = \sum_{b,x \in{0,1}} (-1)^{b\oplus x} \tr[M^0_b\rho^0_x].
		\end{equation}
		
		In terms of $\lambda^0_i, \m^0_i, \asig$ we have, 
		
		\begin{align*}
		\tr[M^0_0\rho^0_0] & = \lambda^0_0(1+m^0_{02})\\
		\tr[M^0_0\rho^0_1] & = \lambda^0_0(1+m^0_{00}\sin 2\theta  +m^0_{02} \cos 2\theta)\\
		\tr[M^0_1\rho^0_0] & = \lambda^0_1(1+m^0_{12})\\
		\tr[M^0_1\rho^0_1] & = \lambda^0_1(1+ m^0_{10}\sin 2\theta + m^0_{12} \cos 2\theta).
		\end{align*}

		In terms of $\lambda^0_i, \m^0_i, \asig$ can rewrite $\Omega^0$ as,
		
		\begin{align}
		\Omega^0 &= \lambda^0_0(1+m^0_{02}) + \lambda^0_1(1+ m^0_{10}\sin 2\theta + m^0_{12} \cos 2\theta) \\ \nonumber
		&- \lambda^0_0(1+m^0_{00}\sin 2\theta  +m^0_{02} \cos 2\theta) - \lambda^0_1(1+m^0_{12}).
		\end{align}
		
		As both $\tr[M^0_0\rho^0_1]$ and $\tr[M^0_1\rho^0_0]$ are positive quantity, hence 
		\begin{equation}
		\label{app_eq:w1}
		\Omega^0 \leq \lambda^0_0(1+m^0_{02}) + \lambda^0_1(1+ m^0_{10}\sin 2\theta + m^0_{12} \cos 2\theta),
		\end{equation}
		and this implies,
		
		\begin{align}
		\label{app_eq:m00}
		(1+m^0_{00}\sin 2\theta  +m^0_{02} \cos 2\theta) & =0\\ \label{app_eq:m12}
		(1+m^0_{12}) &= 0.
		\end{align}
		
		According to the equation \ref{app_eq:m12} we have $m^0_{12} = -1$. As both of $\rho^0_0, \rho^0_1$ lie on the $XZ$ plane and due to the freedom of global unitary without loss of generality we can assume $m^0_{01} = m^0_{11} = m^0_{21} =0$. Due to the positivity constraint ($M^0_i \geq 0$) we have,
		
		\begin{align}
		\label{app_eq:m01}
		{m^0_{00}}^2 + {m^0_{02}}^2 &\leq 1\\ \label{app_eq:m10}
		{m^0_{10}}^2 + {m^0_{12}}^2 &\leq 1\\ \label{app_eq:m20}
		{m^0_{20}}^2 + {m^0_{22}}^2 &\leq 1.
		\end{align}
		
		By combining the constraint equation \ref{app_eq:m12} with the equation \ref{app_eq:m10} we get, $m^0_{10} = 0$. Hence,
		
		\begin{equation}
		\label{app_eq:m}
		\m^0_1 = [0,0,-1],
		\end{equation}
		and by substituting the values of $m^0_{10}, m^0_{12}$ in equation \ref{app_eq:w1} we get the following expression of $\Omega^0$,
		
		\begin{equation}
		\label{app_eq:w2}
		\Omega^0 \leq \lambda^0_0(1+m^0_{02}) + \lambda^0_1(1- \cos 2\theta).
		\end{equation}

		Note that the expression of $\Omega^0$ maximizes when $\lambda^0_0,m^0_{02}, \lambda^0_1$ maximizes and from the constraint equation \ref{app_eq:m01} we get that ${m^0_{00}}^2 + {m^0_{02}}^2 \leq 1$. Hence, without any loss of generality we can assume that for the maximum value of $\Omega^0$, ${m^0_{00}}^2 + {m^0_{02}}^2 = 1$. So, we can parameterize $m^0_{00},m^0_{02}$ as $\sin \alpha, \cos \alpha$ ($0 \leq \alpha \leq 2\pi$). By substituting $m^0_{00} = \sin \alpha, m^0_{02} = \cos \alpha$ in equation \ref{app_eq:m00} we get,
		
		\begin{align*}
		1 + \sin \alpha \sin 2\theta + \cos \alpha \cos 2\theta & =0
		\end{align*}
		This implies, 
		\begin{align*}
		\cos(\alpha - 2\theta) &= -1.
		\end{align*}
		
		As $0 \leq \alpha \leq 2\pi$, so $\cos (\alpha - 2\theta) = -1$this implies,
		
		\begin{align}
		\nonumber
		\alpha - 2\theta &= \pi \quad \text{and,}\\ \label{app_eq:alpha}
		\alpha &= \pi+ 2\theta.
		\end{align}
		
		From the equation \ref{app_eq:alpha} we get,
		
		\begin{equation}
		\m^0_0 = [-\sin 2 \theta, 0, -\cos 2\theta].
		\end{equation}
		
		By substituting the expression of $\m_0$ in equation \ref{app_eq:w2} we get,
		
		\begin{equation}
		\label{app_eq:w3}
		\Omega^0 \leq (\lambda^0_0+\lambda^0_1)(1 - \cos 2\theta).
		\end{equation}

		By substituting the values of $\m^0_0, \m^0_1$ in equation \ref{app_eq:lambm} we get,
		
		\begin{align}
		\label{app_eq:lambb1}
		\lambda^0_2m^0_{22}- \lambda^0_0\cos 2\theta = \lambda^0_1&\\\label{app_eq:lambm1}
		\lambda^0_2m^0_{20} = \lambda^0_0 \sin 2\theta&.
		\end{align}
		
		Due to the constraint equation \ref{app_eq:m20}, similar to $\m^0_0$, here we parameterize the expression of $m^0_{20},m^0_{22}$ as $\sin \beta, \cos \beta$ respectively. By substituting $m^0_{20} = \sin \beta$ and $m^0_{22} = \cos \beta$ in the equations \ref{app_eq:lambb1} and \ref{app_eq:lambm1} we get,
		
		\begin{align}
		\label{app_eq:lamb2}
		\lambda^0_2\cos \beta-\lambda^0_0\cos 2\theta = \lambda^0_1&\\\label{app_eq:lambm2}
		\lambda^0_2\sin \beta = \lambda^0_0 \sin 2\theta&.
		\end{align}

		By solving equation \ref{app_eq:lamb2} and equation \ref{app_eq:lambm2} together with equation \ref{app_eq:lamb} we get,
		
		\begin{align}
		\label{app_eq:lamb0}
		\lambda^0_0 &= \frac{\sin \beta}{\sin \beta + \sin 2\theta + \sin(2\theta - \beta)}\\ \label{app_eq:lamb1}
		\lambda^0_1 &= \frac{\sin (2\theta-\beta)}{\sin \beta + \sin 2\theta + \sin(2\theta - \beta)}.
		\end{align}
		
		Hence,
		
		\begin{align}
		\label{app_eq:lambsum1}
		\lambda^0_0 + \lambda^0_1 &=  \frac{\sin \beta + \sin(2\theta - \beta)}{\sin \beta + \sin 2\theta + \sin(2\theta - \beta)}\\ \label{app_eq:lambsum2}
		&= \frac{\cos (\theta - \beta)}{\cos \theta + \cos (\theta - \beta)}.
		\end{align}

		According to equation \ref{app_eq:w3}, for getting a tight upper bound on $\Omega^0$ we need to maximize $(\lambda^0_0 + \lambda^0_1)$. By equating $\frac{d(\lambda^0_0 + \lambda^0_1)}{d\beta} = 0$ in equation \ref{app_eq:lambsum2} we get,
		
		\begin{equation}
		\frac{\sin(\theta - \beta)\cos \theta}{\cos \theta + \cos (\theta - \beta)} = 0.
		\end{equation}
		
		This implies,
		
		\begin{equation}
		\beta = \theta.
		\end{equation}
		
		It is also easy to check that for $\theta = \beta$, the expression $\frac{d^2 (\lambda^0_0 + \lambda^0_1)}{d \beta^2} < 0$. Hence, the expression $\lambda^0_0 + \lambda^0_1$ maximizes at the point $\beta = \theta$. Substituting this relation in equations \ref{app_eq:lamb0} and \ref{app_eq:lamb1} we get,
		
		\begin{equation}
		\lambda^0_0 = \lambda^0_1 = \frac{1}{2(1+\cos \theta)}.
		\end{equation}

		By substituting the values of $\lambda^0_0+\lambda^0_1$ in equation \ref{app_eq:lamb} we get,
		
		\begin{equation}
		\lambda^0_2 =\frac{\cos \theta}{1 + \cos \theta}.
		\end{equation}
		
		Hence, we get,
		
		\begin{equation}
		\label{app_eq:w4}
		\Omega^0 \leq \frac{2\sin^2 \theta}{1+\cos \theta},
		\end{equation}
		
		and
		
		\begin{align}
		M^0_0 & = \frac{1}{2(1+\cos \theta)} (\id -\sin 2\theta \sigma_X - \cos 2 \theta \sigma_Z)\\
		M^0_1 & = \frac{1}{2(1+\cos \theta)} (\id -\sigma_Z)\\
		M^0_2 & = \frac{\cos \theta}{1+ \cos \theta}( \id + \sin \theta \sigma_X+ \cos \theta \sigma_Z).
		\end{align}
		
		We can rewrite the above expressions as follows,
		
		\begin{align*}
		M^0_0 & = \frac{1}{(1+\cos \theta)} (|1'\rangle \langle1'|)\\
		M^0_1 & = \frac{1}{(1+\cos \theta)} (|1\rangle \langle1|)\\
		M^0_2 & = \id - M^0_0- M^0_1,
		\end{align*}
		
		where $|1'\rangle = \sin \theta |0\rangle - \cos \theta |1\rangle$. This concludes the proof.
		
	\end{proof}
	
	Similarly for the input states $\ket{1}, \ket{1'}$, one can conclude the following.

	\begin{theorem}
		\label{app_thm:max_winD'}
		In POVMgame($M^y,y$), if $A_1$ chooses $y=1$ and the states at $A_2$'s end are $\rho^1_0 = |1\rangle\langle1|$ and $\rho^1_1 = |1'\rangle \langle 1'|$ and if $A_2$ manages to win the game, i.e., $\Omega^1 = \frac{ \sin^2 \theta}{1+\cos \theta}$, then this implies, $A_2$'s measurement devices are of the following form (up to a global unitary),
		
		\begin{align}
		M^1_0 & = \frac{1}{(1+\cos \theta)} (|0'\rangle \langle0'|)\\
		M^1_1 & = \frac{1}{(1+\cos \theta)} (|0\rangle \langle0|)\\
		M^1_2 & = \id - M^1_0- M^1_1,
		\end{align}
		
		where $|0'\rangle = \cos \theta |0\rangle + \sin \theta |1\rangle$.
		
	\end{theorem}

	\begin{proof}
		In the POVMgame($M^y,y$), $A_2$ applies $M^1$ on a single qubit state $\rho^1_x$ (where $x \in_R \{0,1\}$). So, without any loss of generality we can assume that $M^1_i \in M^1$ has the following form,
		
		\begin{equation}
		M^1_i = \lambda^1_i(\id + \m^1_i.\asig),
		\end{equation}
		where $\m^1_i = [m^1_{i0},m^1_{i1},m^1_{i2}]$ and it is the Bloch vector with length at most one, $\asig = [\sigma_X, \sigma_Y, \sigma_Z]$ are the Pauli matrices and $\lambda^1_i \geq 0$. The condition $\sum_{i=0}^2 M^1_i = \id$ leads us to the following relations,
		
		\begin{align}
		\label{app_eq:lamb'}
		\sum_{i=0}^2 \lambda^1_i &= 1\\ \label{app_eq:lambm'}
		\sum_{i=0}^2 \lambda^1_i\m^1_i &= 0.
		\end{align}
		
		In terms of Bloch vector we can rewrite $\rho^1_0, \rho^1_1$ in following way,
		
		\begin{align}
		\rho^1_0 &= \frac{1}{2}(\id - \sigma_Z)\\
		\rho^1_1 &= \frac{1}{2}(\id - \sin 2\theta \sigma_X - \cos 2\theta \sigma_Z).
		\end{align}
		
		In the POVMgame($M^y,y$) if $A_2$ would like to maximizes her winning probability then she needs to maximize the following expression,
		
		\begin{equation}
		\Omega^1 = \sum_{b,x \in{0,1}} (-1)^{b\oplus x} \tr[M^1_b\rho_x].
		\end{equation}
		
		In terms of $\lambda^1_i, \m^1_i, \asig$ we have, 
		
		\begin{align*}
		\tr[M^1_0\rho^1_0] & = \lambda^1_0(1-m^1_{02})\\
		\tr[M^1_0\rho^1_1] & = \lambda^1_0(1-m^1_{00}\sin 2\theta  -m^1_{02} \cos 2\theta)\\
		\tr[M^1_1\rho^1_0] & = \lambda^1_1(1-m^1_{12})\\
		\tr[M^1_1\rho^1_1] & = \lambda^1_1(1- m^1_{10}\sin 2\theta - m^1_{12} \cos 2\theta).
		\end{align*}

		In terms of $\lambda^1_i, \m^1_i, \asig$ can rewrite $\Omega^1$ as,
		
		\begin{align}
		\Omega^1 &= \lambda^1_0(1-m^1_{02}) + \lambda^1_1(1- m^1_{10}\sin 2\theta - m^1_{12} \cos 2\theta) \\ \nonumber
		&- \lambda^1_0(1-m^1_{00}\sin 2\theta -m^1_{02} \cos 2\theta) - \lambda^1_1(1-m^1_{12}).
		\end{align}
		
		As both $\tr[M^1_0\rho^1_1]$ and $\tr[M^1_1\rho^1_0]$ are positive quantity, hence 
		\begin{equation}
		\label{app_eq:w1'}
		\Omega^1 \leq \lambda^1_0(1-m^1_{02}) + \lambda^1_1(1- m^1_{10}\sin 2\theta - m^1_{12} \cos 2\theta),
		\end{equation}
		and this implies,
		
		\begin{align}
		\label{app_eq:m00'}
		(1-m^1_{00}\sin 2\theta  -m^1_{02} \cos 2\theta) & =0\\ \label{app_eq:m12'}
		(1-m^1_{12}) &= 0.
		\end{align}
		
		According to the equation \ref{app_eq:m12'} we have $m^1_{12} = 1$. As both of $\rho^1_0, \rho^1_1$ lie on the $XZ$ plane and due to the freedom of global unitary without loss of generality we can assume $m^1_{01} = m^1_{11} = m^1_{21} =0$. Due to  the positivity constraint ($M^1_i \geq 0$) we have,
		
		\begin{align}
		\label{app_eq:m01'}
		{m^1_{00}}^2 + {m^1_{02}}^2 &\leq 1\\ \label{app_eq:m10'}
		{m^1_{10}}^2 + {m^1_{12}}^2 &\leq 1\\ \label{app_eq:m20'}
		{m^1_{20}}^2 + {m^1_{22}}^2 &\leq 1.
		\end{align}
		
		By combining the constraint equation \ref{app_eq:m12'} with the equation \ref{app_eq:m10'} we get, $m^1_{10} = 0$. Hence,
		
		\begin{equation}
		\label{app_eq:m'}
		\m^1_1 = [0,0,1],
		\end{equation}
		and by substituting the values of $m^1_{10}, m^1_{12}$ in equation \ref{app_eq:w1'} we get the following expression of $\Omega^1$,
		
		\begin{equation}
		\label{app_eq:w2'}
		\Omega^1 \leq \lambda^1_0(1-m^1_{02}) + \lambda^1_1(1- \cos 2\theta).
		\end{equation}

		Note that the expression of $\Omega^1$ maximizes when $\lambda^1_0,\lambda^1_1$ maximizes and $m^1_{02}$ minimizes and from the constraint equation \ref{app_eq:m01} we get that ${m^1_{00}}^2 + {m^1_{02}}^2 \leq 1$. Hence, without any loss of generality we can assume that for the maximum value of $\Omega^1$, ${m^1_{00}}^2 + {m^1_{02}}^2 = 1$. So, we can parameterize $m^1_{00},m^1_{02}$ as $\sin \alpha, \cos \alpha$ ($0 \leq \alpha \leq 2\pi$). By substituting $m^1_{00} = \sin \alpha, m^1_{02} = \cos \alpha$ in equation \ref{app_eq:m00} we get,
		
		\begin{align*}
		1 - \sin \alpha \sin 2\theta - \cos \alpha \cos 2\theta & =0
		\end{align*}
		This implies, 
		\begin{align*}
		\cos(\alpha - 2\theta) &= 1.
		\end{align*}
		
		As $0 \leq \alpha \leq 2\pi$, so $\cos (\alpha - 2\theta) = 1$ this implies,
		
		\begin{align}
		\nonumber
		\alpha - 2\theta &= 0 \quad \text{or} \quad 2\pi \quad \text{and,}\\ \label{app_eq:alpha'}
		\alpha &= 2\theta \quad \text{or} \quad (2\pi + 2\theta).
		\end{align}
		
		One can easily check that for both these values of $\alpha$, the value of $m^1_{00}$ and $m^1_{02}$ are $\sin{2\theta}$ and $\cos{2\theta}$ respectively. From the equation \ref{app_eq:alpha'} we get,
		
		\begin{equation}
		\m^1_0 = [\sin 2 \theta, 0, \cos 2\theta].
		\end{equation}
		
		By substituting the expression of $\m^1_0$ in equation \ref{app_eq:w2'} we get,
		
		\begin{equation}
		\label{app_eq:w3'}
		\Omega^1 \leq (\lambda^1_0+\lambda^1_1)(1 - \cos 2\theta).
		\end{equation}
		
		
		
		By substituting the values of $\m^1_0, \m^1_1$ in equation \ref{app_eq:lambm'} we get,
		
		\begin{align}
		\label{app_eq:lambb'}
		\lambda^1_2m^1_{22}+ \lambda^1_0\cos 2\theta + \lambda^1_1 = 0&\\\label{app_eq:lambm1'}
		\lambda^1_2m^1_{20} + \lambda^1_0 \sin 2\theta = 0&.
		\end{align}
		
		Due to the constraint equation \ref{app_eq:m20'}, similar to $\m^1_0$, here we parameterize the expression of $m^1_{20},m^1_{22}$ as $\sin \beta, \cos \beta$ respectively. By substituting $m^1_{20} = \sin \beta$ and $m^1_{22} = \cos \beta$ in the equations \ref{app_eq:lambb'} and \ref{app_eq:lambm1'} we get,
		
		\begin{align}
		\label{app_eq:lamb2'}
		\lambda^1_2\cos \beta+\lambda^1_0\cos 2\theta + \lambda^1_1 = 0&\\\label{app_eq:lambm2'}
		\lambda^1_2\sin \beta + \lambda^1_0 \sin 2\theta = 0&.
		\end{align}

		By solving equation \ref{app_eq:lamb2'} and equation \ref{app_eq:lambm2'} together with equation \ref{app_eq:lamb'} we get,
		
		\begin{align}
		\label{app_eq:lamb0'}
		\lambda^1_0 &= \frac{\sin \beta}{\sin \beta + \sin(2\theta - \beta) - \sin 2\theta}\\ \label{app_eq:lamb1'}
		\lambda^1_1 &= \frac{\sin (2\theta-\beta)}{\sin \beta + \sin(2\theta - \beta) - \sin 2\theta}.
		\end{align}
		
		Hence,
		
		\begin{align}
		\label{app_eq:lambsum1'}
		\lambda^1_0 + \lambda^1_1 &=  \frac{\sin \beta + \sin(2\theta - \beta)}{\sin \beta + \sin(2\theta - \beta) - \sin 2\theta}\\ \label{app_eq:lambsum2'}
		&= \frac{\cos (\theta - \beta)}{\cos (\theta - \beta)-\cos \theta }.
		\end{align}

		According to equation \ref{app_eq:w3'}, for getting a tight upper bound on $\Omega^1$ we need to maximize $(\lambda^1_0 + \lambda^1_1)$. By equating $\frac{d(\lambda^1_0 + \lambda^1_1)}{d\beta} = 0$ in equation \ref{app_eq:lambsum2'} we get,
		
		\begin{equation}
		\frac{-\sin(\theta - \beta)\cos \theta}{\cos \theta + \cos (\theta - \beta)} = 0.
		\end{equation}
		
		This implies,
		
		\begin{equation}
		\text{either} \quad \beta = \theta \quad \text{or} \quad (\theta-\beta)=\pi.
		\end{equation}
		
		Now, one can easily check that for $\theta = \beta$, the eigen value of $M^1_2$ becomes negative which is not possible. So, the solution here is $(\theta-\beta)=\pi$. One can also check that for $(\theta-\beta)=\pi$, the expression $\frac{d^2 (\lambda^1_0 + \lambda^1_1)}{d \beta^2} < 0$. Hence, the expression $\lambda^1_0 + \lambda^1_1$ maximizes at the point $(\theta-\beta)=\pi$. Substituting this relation in equations \ref{app_eq:lamb0'} and \ref{app_eq:lamb1'} we get,
		
		\begin{equation}
		\lambda^1_0 = \lambda^1_1 = \frac{1}{2(1+\cos \theta)}.
		\end{equation}

		By substituting the values of $\lambda^1_0+\lambda^1_1$ in equation \ref{app_eq:lamb'} we get,
		
		\begin{equation}
		\lambda^1_2 =\frac{\cos \theta}{1 + \cos \theta}.
		\end{equation}
		
		Hence, we get,
		
		\begin{equation}
		\label{app_eq:w4'}
		\Omega^1 \leq \frac{2\sin^2 \theta}{1+\cos \theta}.
		\end{equation}
		
		The corresponding measurement operators using which $A_2$ can achieve $\Omega^1 = \frac{2\sin^2 \theta}{1+\cos \theta}$ is given by,
		
		\begin{align}
		M^1_0 & = \frac{1}{2(1+\cos \theta)} (\id +\sin 2\theta \sigma_X + \cos 2 \theta \sigma_Z)\\
		M^1_1 & = \frac{1}{2(1+\cos \theta)} (\id +\sigma_Z)\\
		M^1_2 & = \frac{\cos \theta}{1+ \cos \theta}(\id- \sin \theta \sigma_X - \cos \theta \sigma_Z).
		\end{align}
		
		We can rewrite the above expressions as follows,
		
		\begin{align*}
		M^1_0 & = \frac{1}{(1+\cos \theta)} (|0'\rangle \langle0'|)\\
		M^1_1 & = \frac{1}{(1+\cos \theta)} (|0\rangle \langle0|)\\
		M^1_2 & = \id - M^1_0- M^1_1,
		\end{align*}
		
		where, $|0'\rangle = \cos \theta |0\rangle + \sin \theta |1\rangle$. This concludes the proof.
		
	\end{proof}

	From the results of theorem \ref{app_thm:max_winD} and \ref{app_thm:max_winD'}, it is clear that the success probability $(1-\cos{\theta})$ in distinguishing two non-orthogonal states $\{\ket{0}, \ket{0'}\}$ (or $\{\ket{1}, \ket{1'}\}$) can be achieved only when the chosen POVM's are of the specified form as chosen by Alice for the QPQ scheme. From the results mentioned in~\cite{iva87}, one can easily conclude that $(1-\cos{\theta})$ is the optimal success probability that can be achieved in distinguishing two non-orthogonal states. So from these two results, one can easily conclude that Alice can get optimal number of raw key bits in this QPQ scheme.
	

		\section*{Appendix D : Correctness of the scheme considering devices ``up to a unitary"}  
		
		In the device independent testing phases of our proposed scheme (i.e., in source device verification phase, Bob's measurement device verification phase and Alice's POVM device verification phase), the tests certify that the devices perform exactly same as that is mentioned in the proposed scheme or ``up to a unitary" of the actual device. This implies that the source device supplies states that are exactly of the same form or ``up to a unitary" (i.e., the states received after applying a unitary operation) of the original state and the measurement devices measure in exactly the same specified basis or ``up to a unitary" (i.e., the measurement bases received after applying a unitary operation) of the actual basis.
		
		Thus, because of this ``up to unitary" deviation, it is necessary to check whether the protocol preserves its correctness condition whenever the devices are ``up to unitary" of the actual devices. 
		
		let us consider that the measurement devices of Alice and Bob perform measurements in the bases which are up to unitary $U_2$ such that
		
		\begin{equation*}
		U_2 = 
		\begin{bmatrix}
		a & b \\
		-e^{i\phi}b^* & e^{i\phi}a^*
		\end{bmatrix}
		\end{equation*} 
		
		where, $a,b \in \mathbb{C}$ such that $|a|^2 + |b|^2 = 1$ and $\phi$ is the relative angle. Let us also assume that the source device supplies states which are up to unitary $U_4$ where
		
		\begin{equation*}
		U_4 = U_2 \otimes U_2
		\end{equation*}
		
		This implies that the states supplied by the source device are of the form
		
		\begin{align*}
		U_4 (\phi_{\A\B}) &= \frac{1}{\sqrt{2}} [\ket{00} + e^{i\phi}(a^*b-ab^*) \ket{01} +\\ 
		~ &~ e^{i\phi}(a^*b-ab^*) \ket{10} + e^{2i\phi}(a^{*^2}+b^{*^2}) \ket{11}]
		\end{align*}
		
		Bob's device measures in the basis $\{U_2\ket{0},U_2\ket{1}\}= \{(a\ket{0}-e^{i\phi}b^*\ket{1}),(b\ket{0}+e^{i\phi}a^*\ket{1})\}$ and $\{U_2\ket{0'},U_2\ket{1'}\}= \{(a\cos{\theta}+b\sin{\theta})\ket{0}+e^{i\phi}(a^*\sin{\theta}-b^*\cos{\theta})\ket{1}, (a\sin{\theta}-b\cos{\theta})\ket{0}-e^{i\phi}(a^*\cos{\theta}+b^*\sin{\theta})\ket{1}\}$ instead of the basis $\{\ket{0}, \ket{1}\}$ and $\{\ket{0'}, \ket{1'}\}$ respectively. Alice's POVM devices are either $D'^0=\{D'^0_0,D'^0_1,D'^0_2\}$ or $D'^1=\{D'^1_0,D'^1_1,D'^1_2\}$ for $a_i=0$ and $a_i=1$ respectively where
		
		\begin{align*}
		D'^0_0 & = \frac{1}{(1+\cos \theta)} (U_2\ket{1'}\bra{1'}U_2^{\dagger})\\
		D'^0_1 & = \frac{1}{(1+\cos \theta)} (U_2\ket{1}\bra{1}U_2^{\dagger})\\
		D'^0_2 & = \id - D'^0_0- D'^0_1,
		\end{align*}
		
		and
		
		\begin{align*}
		D'^1_0 & = \frac{1}{(1+\cos \theta)} (U_2\ket{0'}\bra{0'}U_2^{\dagger})\\
		D'^1_1 & = \frac{1}{(1+\cos \theta)} (U_2\ket{0}\bra{0}U_2^{\dagger})\\
		D'^1_2 & = \id - D'^1_0- D'^1_1,
		\end{align*}

		One can easily check that whenever Bob measures in $\{U_2\ket{0},U_2\ket{1}\}$ or $\{U_2\ket{0'},U_2\ket{1'}\}$ basis randomly on his qubit of the shared state $U_4 (\phi_{\A\B})$, the qubit at Alice's side will also collapse to $U_2\ket{0}$ or $U_2\ket{1}$ for the first case and $U_2\ket{0'}$ or $U_2\ket{1'}$ for the second case.
		
		Now, if Alice chooses POVM device $D'^0=\{D'^0_0,D'^0_1,D'^0_2\}$ for $a_i=0$, the probabilities of getting different outcomes for two different input states are as follows- 
		
		\begin{eqnarray*}
			\Pr(D'^0_0|U_2\ket{0}) & =& (1 - \cos{\theta})\\
			\Pr(D'^0_1|U_2\ket{0}) &=& 0\\
			\Pr(D'^0_2|U_2\ket{0}) &=& \cos{\theta}\\
			\Pr(D'^0_0|U_2\ket{0'}) & =& 0\\
			\Pr(D'^0_1|U_2\ket{0'}) &=& (1 - \cos{\theta})\\
			\Pr(D'^0_2|U_2\ket{0'}) &=& \cos{\theta}\\
		\end{eqnarray*}
		
		Similarly, if Alice chooses POVM device $D'^1=\{D'^1_0,D'^1_1,D'^1_2\}$ for $a_i=1$, the probabilities of getting different outcomes for two different input states are as follows-
		
		\begin{eqnarray*}
			\Pr(D'^1_0|U_2\ket{1}) & =& (1 - \cos{\theta})\\
			\Pr(D'^1_1|U_2\ket{1}) &=& 0\\
			\Pr(D'^1_2|U_2\ket{1}) &=& \cos{\theta}\\
			\Pr(D'^1_0|U_2\ket{1'}) & =& 0\\
			\Pr(D'^1_1|U_2\ket{1'}) &=& (1 - \cos{\theta})\\
			\Pr(D'^1_2|U_2\ket{1'}) &=& \cos{\theta}\\
		\end{eqnarray*}
		
		According to the protocol, whenever $a_i=0$ and 
		Alice gets $D'^0_0(D'^0_1)$, she outputs $r_{\A_i}=0 (1)$. Whenever, $a_i = 1$ and she gets $D'^1_0(D'^1_1)$, she outputs $r_{\A_i}=0 (1)$. So, in this case, the success probability of Alice to guess the $i$-th raw key bit $r_i$ of Bob will be,
		
		\begin{eqnarray}
		& &\Pr(r_{\A_i} = r_i) \nonumber\\
		&=&\Pr(r_{\A_i}=0,r_i=0) + \Pr(r_{\A_i}=1,r_i=1) \nonumber\\
		&=& (1 - \cos{\theta}).\nonumber
		\end{eqnarray}
		
		This shows that whenever the devices (both source and measurement devices) involved in this scheme are ``up to a unitary" of the original specified device, then also the proposed scheme satisfies the correctness condition.

\end{document}